\documentclass[compsoc,conference,a4paper,10pt,times]{IEEEtran}
\IEEEoverridecommandlockouts

\newif\ifdraft
\newif\iffull
\newif\ifcamera
\newif\ifcomments

\commentstrue
\fullfalse
\cameratrue
\draftfalse

\usepackage{amsmath}
\usepackage{balance}
\usepackage{amsthm}
\usepackage{inconsolata}
\usepackage[inline,shortlabels]{enumitem}
\usepackage{xspace}
\usepackage{cite}
\usepackage{graphicx}
\usepackage{algorithmicx}
\usepackage{algpseudocode}
\usepackage{algorithm}
\usepackage{xcolor}
\usepackage{multirow}
\usepackage{hhline}
\usepackage{listings}
\usepackage{tabularx, booktabs}
\usepackage{microtype}
\newcolumntype{C}[1]{>{\hsize=#1\hsize\centering\arraybackslash}X}%
\usepackage[hyphens]{url}
\usepackage[colorlinks=true,
            linkcolor=black,
            urlcolor=blue,
            citecolor=blue,
            letterpaper,
            breaklinks,
            pdfborder={0 0 0}]{hyperref}
\usepackage[capitalize,nameinlink]{cleveref} %
\usepackage{makecell}

\ifdraft
    \usepackage{fancyhdr, datetime}
\fi

\renewcommand{\paragraph}[1]{\smallskip\noindent\textbf{#1.}\enskip}

\definecolor{FGreen}{cmyk}{0.9,0.2,0.5,0.3}

\newcommand{\eg}{\textit{e.g.},\@\xspace}
\newcommand{\ie}{\textit{i.e.},\@\xspace}
\makeatletter
\newcommand*{\etc}{%
    \@ifnextchar{.}%
        {etc}%
        {etc.\@\xspace}%
}
\makeatother
\makeatletter
\newcommand*{\etal}{%
    \@ifnextchar{.}%
        {\textit{et al}}%
        {\textit{et al.}\@\xspace}%
}
\makeatother

\def\txtapprox{$\mathtt{\sim}$} %

\newtheorem{claim}{Claim}

\crefformat{section}{\S#2#1#3}

\makeatletter
\newtheorem*{rep@theorem}{\rep@title}
\newcommand{\newreptheorem}[2]{%
\newenvironment{rep#1}[1]{%
 \def\rep@title{#2 \ref{##1}}%
 \begin{rep@theorem}}%
 {\end{rep@theorem}}}
\makeatother
\newreptheorem{claim}{Claim}

\def\BibTeX{{\rm B\kern-.05em{\sc i\kern-.025em b}\kern-.08em
    T\kern-.1667em\lower.7ex\hbox{E}\kern-.125emX}}

\begin{document}
\title{Practical Volume-Based Attacks on Encrypted Databases}

\author{\IEEEauthorblockN{Rishabh Poddar$^*$
\thanks{$^*$Both authors contributed equally (ordered alphabetically).}}
\IEEEauthorblockA{\textit{UC Berkeley} \\
rishabhp@eecs.berkeley.edu}
\and
\IEEEauthorblockN{Stephanie Wang$^*$}
\IEEEauthorblockA{\textit{UC Berkeley} \\
swang@cs.berkeley.edu}
\and
\IEEEauthorblockN{Jianan Lu}
\IEEEauthorblockA{\textit{Princeton University} \\
jiananl@princeton.edu}
\and
\IEEEauthorblockN{Raluca Ada Popa}
\IEEEauthorblockA{\textit{UC Berkeley} \\
raluca@eecs.berkeley.edu}
}

\maketitle

\begin{abstract}
Recent years have seen an increased interest towards strong security primitives for encrypted databases (such as oblivious protocols),
that hide the access patterns of query execution, and reveal only the volume of results.
However, recent work has  shown that even volume leakage can enable the reconstruction of entire columns in the database. %
Yet, existing attacks rely on a set of assumptions that are unrealistic in practice: for example, they
\begin{enumerate*}[(i)]
    \item require a large number of queries to be issued by the user, or 
    \item assume certain distributions on the queries or underlying data (\eg that the queries are distributed uniformly at random, or that the database does not contain missing values).
\end{enumerate*}

In this work, we present new attacks for recovering the content of individual user queries, 
assuming no leakage from the system except the number of results and avoiding the limiting assumptions above. 
Unlike prior attacks, our attacks require only a {\em single} query to be issued by the user for recovering the keyword. Furthermore, our attacks make no assumptions about the distribution of issued queries or the underlying data. Instead, our key  insight is to exploit the \emph{behavior} of real-world applications.

We start by surveying 11 applications to identify two key characteristics that can be exploited by attackers---(i)~file injection, and (ii)~automatic query replay. We present attacks that leverage these two properties in concert with volume leakage, independent of the details of any encrypted database system. Subsequently, we perform an attack on the real Gmail web client by simulating a server-side adversary. Our attack on Gmail completes within a matter of minutes, demonstrating the feasibility of our techniques. We also present three ancillary attacks for situations when certain mitigation strategies are employed.
\end{abstract}

\section{Introduction}
\label{sec:intro}

In recent years, there has been a tremendous increase in interest towards encrypted database systems that enable queries over encrypted data,
because they provide privacy guarantees against a compromised database server.
A number of practical systems have been proposed by academia as well as industry \cite{Fuller:EDBs, CashJJJKRS14, FJKNRS15:RichQueries, popa:cryptdb, AEKKRV13:cipherbase2, TKMZ13:monomi, website:MSalwaysEncryptedDB, website:skyhigh, website:ciphercloud, website:iqrypt, poddar:arx}, typically relying on techniques such as property-preserving encryption \cite{boldyreva:ope, boldyreva:ope-revisited, popa:mope, KS14:optimalAvgOPE, WuLewiRange} or searchable encryption \cite{SongWP00, CurtmolaGKO06, KamaraPR12, CashJJKRS13, OgataKKM13, CashJJJKRS14, LauCSJLB14, Kurosawa14, NaveedPG14,StefanovPS14, HeAJSS14, Bost16}.

Most of these schemes leak {\em query access patterns}. 
Consider the example of an email application: a user issues a search query for a keyword over their emails. To facilitate such queries, the mail server typically stores an inverted index (also called a {\em secondary index}) for each user's mailbox, which maps each keyword to the list of emails it appears in.
When fetching the results of a queried keyword, a compromised server can observe the set of email identifiers that match the keyword (\ie the access patterns of the query), even though the email bodies remain encrypted.
 A set of recent works~\cite{IslamKK12, CashGPR15, ZhangKP16, LiuZWT14, AbdelraheemAG17, GrubbsMNRS16, PouliotW16, GiaruadABL17,Islam:2014:rangepatterns,KellarisKNO16,Dautrich:2013,KennyAccessPattern, GLMP:attack:2019, KPT:knnattack:2019, KPT:attack:2020} has shown that such access patterns leak significant information to the attacker, enabling the identification of keywords that users search for as well as email contents.

Many of these works discuss oblivious protocols 
such as ORAM (Oblivious RAM)~\cite{GoldreichO96,StefanovDSFRYD13} or PIR (Private Information Retrieval)~\cite{PIR:survey} as a solution to this leakage. 
Even an attacker eavesdropping at the server is unable to identify which documents were returned in response to a query (\ie the access patterns of queries remain hidden).
Instead, the attacker can only observe the {\em volume} of results.
Consequently, these schemes are often regarded as conferring a very strong security guarantee, the main downside largely being their slow performance.

However, in seminal work, Kellaris \etal~\cite{KellarisKNO16} showed that even schemes that provably conceal access patterns allow attackers to reconstruct the {\em database counts}, \ie the number of documents in the database containing each particular value. 
The attacker neither knows the content of individual queries (which are encrypted), nor which documents were returned in response. %
Instead, he only observes the  volume of results given a set of range queries.
 Kellaris \etal~showed that volume-based attacks were possible, even if not yet practical. Their techniques required the attacker to observe the result volumes of $O(N^4\log N$) range queries, for a domain of size $N$. Furthermore, they also assumed that the range queries were drawn at random from a uniform distribution, thus severely limiting the applicability of the attack in practice.
Grubbs \etal~\cite{Grubbs18:Volume} and (more recently) Gui \etal ~\cite{Gui19:Volume} improved upon the results of Kellaris \etal by presenting attacks that do not require a uniformity assumption for queries
as long as other assumptions hold---\eg that queries with all possible volumes (within a certain bound) are observed at least once; or that the underlying database is dense (\ie all $N$ values occur in the database).

In this work, we %
explore an alternative design point in the space of attacks, and
show that volume-based attacks are practical {\em without making any assumptions about the distribution of queries or the underlying data}.
Our aim is to recover the content of individual queries that search for a specific keyword in the database. 
We note that as long as a query for each keyword is issued at least once, our attack enables an adversary to reconstruct the list of all the keywords that appear in the database.

In particular, we focus on the behavior of applications that allow users to search for keywords over a {\em secondary index}, a common data structure in database systems that maps keys to a set of matching records.
In the encrypted database literature, this corresponds to the model of searchable encryption schemes~\cite{SongWP00, CurtmolaGKO06, KamaraPR12, CashJJKRS13, OgataKKM13, CashJJJKRS14, LauCSJLB14, Kurosawa14, NaveedPG14, StefanovPS14, HeAJSS14, Bost16}.

Our key insight is that by exploiting the \emph{behavior} of specific real-world applications, we can avoid assumptions made by prior volume attacks about the distribution of queries or the underlying data.
Furthermore, it allows our attack to be eminently practical, requiring only a {\em single} query to be issued by the user for recovering the keyword. %

As such, even though real-world applications today leak far more information than just the volume of results, the importance of volume attacks will only grow in the future. Privacy-conscious services have begun deploying sophisticated schemes to plug traditional sources of leakage, including access patterns (\eg the Signal messaging service~\cite{Signal, SignalDiscoveryBlog}).
The takeaway of our work is that as practitioners take steps towards enhancing the privacy guarantees of their applications, they must also account for the leakage of result volumes. %
Application-specific behavior that facilitates easy exploitation of this leakage should be revised.

\subsection{Techniques and contributions}
We start by examining 11 representative applications that enable search queries over a secondary index (\eg Gmail, Twitter, and Facebook) to identify realistic attacker capabilities that can be leveraged in concert with volume leakage. 
We find that many of these applications satisfy two key characteristics that enable us to mount efficient volume-based attacks,
even if they are built atop a cryptographic backend (such as ORAM or PIR) that plugs traditional sources of leakage, and only leaks the volume of results.

First, we find that many applications inherently allow {\em other} users to {\em inject} application data into a victim user's index. This property of applications has also been noted in prior work~\cite{CashJJJKRS14,LiuZWT14,ZhangKP16}. In our setting, it allows the attacker to potentially influence the volume of results returned by a query.

However, it is not clear how to leverage file injection alone when the only information available is the number of results.
File-injection attacks have been studied in the searchable encryption literature~\cite{CashJJJKRS14,LiuZWT14,ZhangKP16}, but these attacks rely crucially on the attacker knowing the query access patterns---namely, the set of files matching the keyword.
The key idea is that the attacker injects special files $F_1, \dots, F_n$ constructed so that each keyword is contained in a unique subset of files.
When the victim queries for a keyword, the attacker learns the exact set of files returned and hence the keyword.
In our model, though, the attacker only learns the number of files returned, not the exact set.
Ensuring that each keyword has a unique number of files means injecting $|D|$ files, where $D$ is the dictionary space.
For the English dictionary (\txtapprox $ 200K$ words), this attack would only be feasible on a small subset.
Moreover, given a set of files, there may be many different keywords that are present in the same number of files, precluding these attacks in our setting.

Instead, our strategy is to leverage a second property we observed in the real applications surveyed, which is key for making an injection attack feasible with only volume leakage:
the ability to {\em replay} queries issued by a user {\em without further user intervention}.
While this seems at first glance to be a strong assumption, we find that several applications display this behavior as a {\em built-in feature}, ostensibly to hide transient application errors.

As an example, Gmail inbox search fits our setting seamlessly, and satisfies both the aforementioned properties.
The user can search for a keyword in their emails, an attacker can inject data by simply sending the user an email with a specific keyword, and the user's query is automatically replayed by the application when the server's response is delayed without relying on user intervention.
In~\cref{s:abilities}, we define these abilities formally, demonstrate how they appear more generally in a wide range of applications, and explain why they are hard to avoid.

Given these attacker abilities, we present attack algorithms that are able to reconstruct user queries 
on secondary indices. %
Our high-level strategy, described in \cref{s:attack}, is twofold:

\begin{enumerate}
\item {\bf Inject:} the attacker injects $k$ {\em specially crafted} files that alter the number of results for a candidate query. %
\item {\bf Replay:} the attacker causes the client to automatically replay the query without user intervention.
\end{enumerate}

The process repeats to narrow the search space, without the user's knowledge. This base attack succeeds with $100\%$ probability in identifying words in the attacker's dictionary.
Specifically, given a dictionary of keywords $D$ that represents the attacker's domain of interest,
the attack recovers a queried keyword in only $O(\log_k |D|)$ replays of the query (\eg $<5$ for the English dictionary).

Subsequently, we build upon our base attack to present three ancillary attacks. The ancillary attacks show that our strategy remains feasible, albeit more expensive or less accurate, even when certain mitigation techniques are employed (\cref{sec:mitigate}).
\begin{itemize}
\item \textbf{No replay attack:} We provide an extension to our base attack that works even when query replay is not possible~(\cref{s:attack:noreplay}). While the detection accuracy decreases in this case due to the mentioned shared cardinality, the attack still succeeds with significant probability on a smaller dictionary.
\item \textbf{Attack with padding:} We further demonstrate that our attack remains possible even when padding is used to hide result counts (\cref{s:attack:padding}), through an extension to our base attack that requires more injected data but with proven effectiveness.
\item \textbf{Attack with noise:} We finally show an extension of the base attack with noisy data when the attacker does not know the result count precisely (\cref{s:attack:noise}).
\end{itemize}

In \cref{s:conjunctions}, we also discuss extensions to the attack for recovering keywords in conjunctive queries.

Finally, we demonstrate the feasibility of our techniques and attack the real Gmail web client by simulating a server-side adversary (\cref{s:eval:gmail}). The characteristics of real applications poses a number of constraints on the attack, \eg based on the behavior of replays, or the time it takes to inject files into the secondary index. Despite these constraints, we show that our attack completes within a matter of minutes for the Gmail application.
We also analyze the theoretical complexity of our attacks, and experimentally evaluate their overheads and accuracy in a variety of settings.

While there are ways to mitigate the attack, the vulnerability from result count leakage is difficult to eradicate from a system completely.
Padding to the worst-case count theoretically prevents leakage, but in many applications, this results in unaffordable overheads.
Rather than trying to reduce system leakage, we believe that the most effective mitigation is actually on the {\em application} side, although these techniques too may be burdensome because they interfere with application-specific functionality (\eg disallowing users from sending email).
In \cref{sec:mitigate}, we discuss these mitigations, but we note that in general, it is difficult to protect against the attack completely because it relies on very little from the system model, and instead exploits features inherent to the application model.

\section{Attack model}
\label{sec:setup}

In this section, we discuss the generic system model~(\cref{sec:setup:system}) and the application model~(\cref{sec:setup:application}) that is vulnerable to our attacks.
We present the three key attack assumptions:
\begin{enumerate*}
\item that the system leaks volume,
\item that the application allows data injection, and
\item that the applications automatically replays queries under certain scenarios.
\end{enumerate*}
We then demonstrate the validity of our assumptions by studying a number of concrete instances for both the system and application models. In particular, we examine 11 popular web applications that allow users to issue keyword search queries over an inverted index---we find that 
\begin{enumerate*}[label=(\roman*)]
\item all 11 applications allow attackers to inject data into the victim user's index; and 
\item 5 of the 11 applications also replay queries automatically without user intervention.
\end{enumerate*}

\subsection{System Model}
\label{sec:setup:system}

We consider systems in which an untrusted server (the adversary in our setting) maintains a secondary index in an encrypted database.
The index maps a keyword to a list of documents or database rows (referred to as {\em files}, henceforth) that the keyword appears in and is stored on the server for query efficiency.
Whenever the application proxy or the client queries the index for a keyword,
the user receives the corresponding list of files containing the keyword.
We assume that the query's execution reveals no information to the adversary except the number of results.

More formally, %
we define a database
$\mathcal{D}$ as a set of records that associates keywords with the files from a collection $\mathcal{F}$ that the keywords
appear in:
\[ \mathcal{D} = \left\{(w, f) : w \in f, f \in \mathcal{F} \right\} \]

A query for word $w$ is a function $q_w$ (where $w$ is private) that maps
$\mathcal{D}$ to a list of matching files in $\mathcal{F}$:
\[ q_w(\mathcal{D}) = \left\{ f : (w, f) \in \mathcal{D} \right\}. \]

An implementation of such a database may internally use one layer of
indirection, so that the first query returns a list of file pointers, or
indices into $\mathcal{F}$, and the subsequent queries are used to fetch the
file contents from $\mathcal{F}$.

The adversary's goal is to identify the keyword $w$ using only the size of the result set $|q_w(\mathcal{D})|$.

\begin{figure}[t]
\centering
\includegraphics[width=\linewidth]{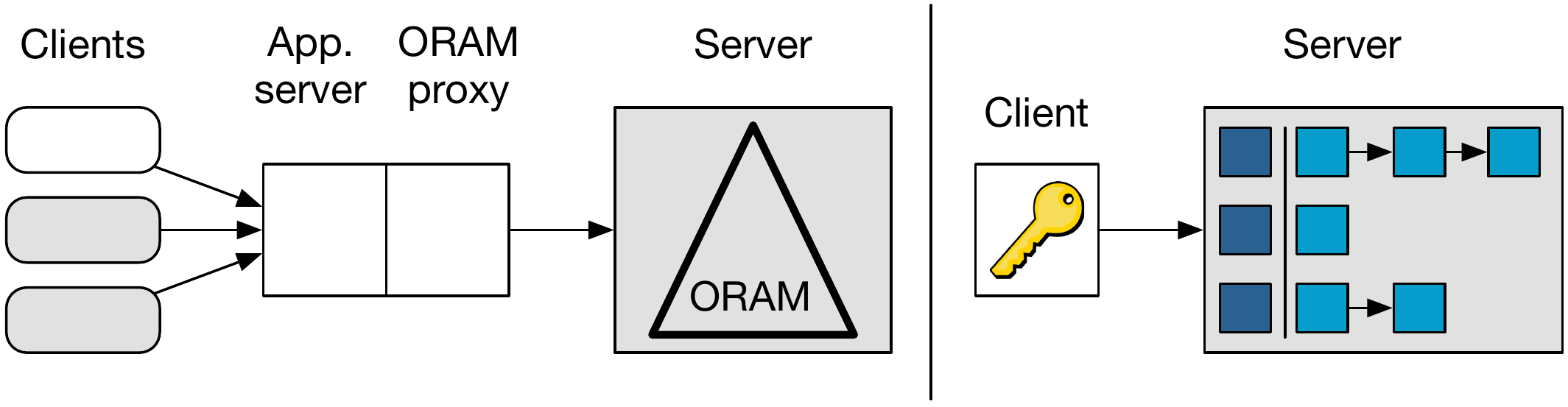}
\caption{(Left) ORAM model: the server and a subset of the
clients are untrusted (shaded). (Right) PIR model: the server is untrusted.} 
\label{fig:oram-model}
\label{fig:pir-model}
\end{figure}

\paragraph{Examples}
We now illustrate the relevance of volume-based attacks by discussing concrete examples of cryptographic systems that leak the volume of results to attackers.
We consider a client-server model where the database stored at the server is encrypted using sophisticated techniques that also hide access patterns, and the server maintains a secondary index over the encrypted data.
As we'll see in the following examples, content encryption is not sufficient to prevent volume leakage.

\textit{ORAM-based systems.}\enskip
\label{s:instances:email} \label{s:instances:oram}
In the ORAM model~(\cref{fig:oram-model}, left), the database administrator additionally backs the database
with ORAM, hiding access patterns from the server in addition to the result
contents.
However, even with a guarantee of this strength, volume leakage is possible for a passive attacker because the \emph{size} of the result contents is not hidden.
In addition, even if a layer of indirection is used, so that the first query only returns a list of file pointers, the number of files returned can still be measured by recording the number of subsequent queries made.

\iffull
In such database-backed applications, secondary indices are a common
optimization in which an index is built on a
field that is not the main (``primary'') key in that table, often not unique.
The email application is a special case of this scenario. A user is looking for emails that match a
specified keyword. Since the search could include a large fraction of the
user's inbox, it is generally performed using a reverse index stored on some
backend storage server, \eg an ORAM node, and the results are returned to the
user over a secure connection. In this case, there is a secondary index per
user, each of which can be updated by all other users.
\fi

\textit{PIR-based systems.}\enskip
\label{s:instances:pir}
In the PIR model (\cref{fig:pir-model}, right), an untrusted server owns the
database and maintains a secondary index on it for fast access. In this case,
the server sees all the data, as its role is to maintain and serve publicly
available data.
The untrusted server answers user queries in which the key requested is
private~~\cite{PIR:survey}.
However, since the data is publicly available, the untrusted server can easily learn the size of the query results.

\subsection{Application model}
\label{sec:setup:application}
\label{s:abilities}

\begin{figure*}[t!] %
\small
\begin{center}
    \begin{tabular}{lllc}
        \toprule
        \thead[l]{Application} & \thead[l]{Type of queries} & \thead[l]{File injection strategy} & \thead{Number of replays \\ in 10 minutes} \\\midrule
        Gmail & Email keywords & Send emails to victim & 6 \\ %
        Facebook & Names, post keywords & Create posts in a group page & 4 \\ %
        Dropbox & File keywords & Upload files to a shared folder & 1 \\ %
        Google Doc & File keywords & Upload files to a shared folder & 1 \\ %
        iCloud Mail & Email keywords & Send emails to victim & 1 \\ %
        Twitter & Hashtags & Post tweets with hashtags & 0 \\ %
        Piazza & Post keywords & Create posts in a class & 0 \\%\hline
        Slack & Names, message keywords & Send messages to a group channel & 0 \\%\hline
        Skype & Names, message keywords & Send messages to victim & 0 \\%\hline
        Yahoo Mail & Email keywords & Send emails to victim & 0 \\%\hline
        Outlook Mail (Hotmail) & Email keywords & Send emails to victim & 0 \\%\hline
        \bottomrule %
    \end{tabular}
\caption{%
        An empirical assessment of 11 popular web applications. In each case, we list the type of query made by the user, how the attacker can influence the result of this query by application-specific injection, and the number of replays observed. To measure replay, all server responses are dropped for 10 minutes, and we report the number of duplicate queries made within 10 minutes.}
    \label{fig:apps}
\end{center}
\end{figure*}

We assume an application model based on the behavior of actual applications that rely on a secondary index.
The model consists of two key assumptions: the ability to inject data into a user's index and the ability to replay a user query without the involvement of the user.
We argue that these two assumptions are in fact often inherent to the application.
As evidence, we survey a variety of popular web applications that rely on a secondary index and find that both assumptions hold in 5 out of the 11 surveyed.
Finally, we describe how these assumptions together make it difficult to detect our attack.

\subsubsection{File injection}
First, we consider the assumption that the attacker can inject data into a user's index.
For applications that involve interactions between multiple users, injection by other users is often necessary for application functionality.
For example, to inject into Gmail inbox search, the attacker sends the victim user an email. 
Injection is especially easy if
there is a secondary index that is shared. To inject an entry for a hashtag in Twitter, 
the attacker uploads a post with that hashtag; in Slack, the attacker simply sends a message. 
The ability of the attacker to inject in such applications is fundamental
because these applications are inherently designed for multiple users and
contain shared data. 

\subsubsection{Query replay}
Second, we assume that the attacker can replay a user's query a finite number of times.
While this assumption is certainly not universal across applications that rely on a secondary index, we find that it is surprisingly common, with many applications replaying queries automatically in the background without any user intervention.
This is because many applications are written to handle transient errors transparently, to put as little burden on the user as possible.
In particular, an application that wants to provide a seamless experience when network connectivity is spotty may retry a query automatically if the response is too slow.
Indeed, the HTTP/1.1 RFC~\cite{httprfc} specifies that
``When an inbound connection is closed prematurely, a client MAY open a new
connection and automatically retransmit an aborted sequence of [idempotent]
requests.''
A compromised server can force this behavior by simply dropping its HTTPS responses, triggering an automatic replay.

\paragraph{Examples} To show that these assumptions are realistic, we surveyed 11 applications, including Gmail, Twitter, and Facebook, and tested for the ability to inject data and replay queries for a target user~(\cref{fig:apps}).
To test for injection, we examined the application functionality to determine whether the attacker could inject data into an index searchable by the user.
To measure the number of query replays, we drop responses from the server, record all network traffic from the application client, and count the number of duplicate queries that appear within 10 minutes.
We find that for all applications, injection is possible, although sometimes only if the attacker and the user share some index (\eg they are both members of a public Facebook group).
We also find that 5 of these applications have query replay.

On further investigation of these 5 applications, we find that all retry queries automatically, though the rate of retries varies.
Two applications, Gmail and Facebook, retry the query repeatedly.
The remaining three---Dropbox, Google Drive, and iCloud Mail---retry the query once.
The number of retries is important because more retries make it easier for the attacker to identify the query.
Nevertheless, as we show in
\cref{s:eval:basic}, even a single replay is sufficient for significantly pruning the space of query
possibilities, and, in many cases, for mounting the attack feasibly.

Some applications do not replay a query automatically, as is the case with Twitter or Slack.
In~\cref{s:attack:noreplay} we provide a single-round version of our
attack that does not require queries to be replayed at all. 
This version of our
attack is predictably less effective than the base attack with the ability to
replay, but as we show in~\cref{s:eval:noreplay}, it is 
practical for small attacker dictionaries.

\paragraph{Avoiding detection} Because the attack relies on injection visible to the user, one practical concern in launching the attack is avoiding detection.
Fortunately, in many settings, our attack is {\em difficult to detect before  it completes}.
This is because once the user issues a query, the attacker can continue
to drop / block the responses to the client, causing the application to retry queries until the attack completes.

We verified this behavior with Gmail: no results are returned to the user during
the attack, and to the user it simply appears that they have a bad network connection. 
That is, once the user initiates the query, the attack will complete without
further actions from the victim user.

It is possible that the user later sees the injected emails and realizes from the
synthetic content that they are
under attack, but this happens only {\em after the attack completes}.
Further, we note that although services like Gmail may strip suspicious HTML
elements during email preprocessing, we can still use style formatting to avoid
showing the injected content to the user, to reduce suspicion. The rest of the email could show
content that is more user-friendly, \eg an ad.
It is further unlikely for spam filters to detect the injected emails, since the
attack targets a specific user.
This is just one example, but it illustrates the numerous ways that an attacker could inject data in a way that is difficult to detect before the query is reconstructed.

\section{Attacks}
\label{s:attack}

Given the attacker abilities discussed in \cref{sec:setup}, we present and
analyze a file-injection attack to recover a user's query on a secondary index.
We show that this attack can be launched on the generic database described in
\cref{sec:setup:system}, as long as the attacker can view the number of
results returned. This is true even if the result content is encrypted and a
model like ORAM is used to hide access patterns.

The general attack (\cref{s:attack:basic}) can recover a user's query with
100\% accuracy, by leveraging the three assumptions presented in
\cref{sec:setup}. One may attempt to  weaken the attacker's abilities
by making it more difficult to replay a user's query, or padding the result sets.  
We discuss extensions to the attack in
\cref{s:attack:noreplay}--\cref{s:attack:noise} and show that it is still
feasible even when various countermeasures are employed, albeit at higher
overheads and possibly imperfect accuracy. 
\cref{fig:complexity} summarizes the overheads of the base attack and its extensions.
In \cref{s:conjunctions}, we describe further extensions to the attack for recovering keywords in conjunctive queries.

\subsection{Base Attack}
\label{s:attack:basic}

At a high level, the base attack works by searching on the keyword
universe through multiple rounds of user query replay.  By recording the
result counts between rounds, the attacker can narrow down the keyword search
space by a constant factor per round.

The attacker uses file injection to influence the result count between rounds.
During each round, the attacker constructs files from the keyword search space
and injects the files into the user's index. The response for each round will
then contain some number of injected files. The attacker can use the new
result count to determine the number of files injected after the
previous round. In this way, the attacker can determine which subset of the
search space contains the user's query.

The setup of the attack is as follows: A user queries $q_w$ on a database
$\mathcal{D}$, as defined in \cref{sec:setup}. The response is the set of
matching file contents, $q_w(\mathcal{D}) = \{ f_1, \dots, f_n \}$. The goal of
the attack is to recover $w$, using only $n = |q_w(\mathcal{D})|$, the number 
of files returned.

Algorithm~\ref{alg:attack} provides pseudocode for the base attack \textsc{RecoverQuery}.
In more detail, the attacker first records the user
query's $q_w$ and the number of files returned, $n_0$. $n_0$ is the number of
files that already matched to $w$ prior to the attack.  This enables the
attacker to differentiate user-uploaded files from injected ones.

Next, the attacker proceeds in rounds to reduce the keyword search space. 
He chooses an initial dictionary $D_0$, a set of words that
might contain $w$, and a parameter $k$. During each round $j$, the attacker divides $D_j$
into $k$ equal partitions. He injects $k$ files into the database and distributes the words
among them as follows: \emph{If a word appears in the $i$-th partition, he
adds the word to exactly $i$ out of the $k$ files.} Hence, if a word appears in the $k$-th partition, 
the attacker adds this word to all $k$ files.

The attacker then replays the user's query $q_w$ on the updated database and
records the number of files returned, $n_j$. Assuming that the attacker can
block updates to the secondary index, the number of files injected since the
previous round is then $i^* = n_j - n_{j-1}$. Thus, $w$ must have been assigned
to the $i^*$-th partition during round $j$. The attacker repeats this in rounds,
each time using the $i^*$-th partition as the new dictionary, until $|D| = 1$.

\begin{algorithm}[t]
\small
    \caption{
        Pseudocode for the base attack. %
        $q_w$ is a private query for a word $w$ on a database $\mathcal{D}$.
        Each round of the attack partitions the search space by $k$.
    }\label{alg:attack}
    \begin{algorithmic}[1]
        \Procedure{\textbf{RecoverQuery}}{$q_w$, $k$}
        \State $\mathcal{D} \gets$ the initial database
        \State $D \gets$ keyword universe
        \State $n \gets |q_w(\mathcal{D})|$
        \While{$|D| > 1$}
            \For{$i$ in $[1, \dots, k]$}
                \State $F_i \gets$ an empty file
                \State $D_i \gets$ an empty dictionary
            \EndFor
            \For{$index$ in $[1, \dots, |D|]$}
                \State $w \gets D[index]$
                \State $i \gets \left\lfloor \frac{index}{{|D|} / {k}} \right\rfloor$
                \State Append $w$ to $i$ unique files in $F$
                \State Add $w$ to dictionary $D_i$
            \EndFor
            \State $\mathcal{D} \gets \textsc{InjectFiles}(\mathcal{D}, F)$
            \State $n' \gets |q_w(\mathcal{D})|$
            \State $i \gets n' - n$
            \State $D \gets D_i$
            \State $n \gets n'$
        \EndWhile
        \State \Return $D[0]$
        \EndProcedure
    \end{algorithmic}
\end{algorithm}

The complete overheads for the attack are summarized in
\cref{fig:complexity}. This attack converges in a bounded
number of rounds since each round is guaranteed to reduce the dictionary size.
Furthermore, for a high enough $k$ and a small enough $D$, the number of
rounds, \ie the number of times the attacker has to replay the user's query, is
quite low. We formalize this in the following claim:

\begin{claim}
    \label{claim:numrounds}
    For any dictionary $D$ and for any word $w \in D$, let
    $q_w$ be a private query for $w$, and $k$ be the number of partitions. Then,
    \emph{\textsc{RecoverQuery}}$(q_w, k)$ returns $w$ after $\left\lceil \log_k{|D|}
    \right\rceil$ rounds.
\end{claim}
\begin{proof}
    Consider the $j$-th round of the attack, which searches a dictionary
    $D_j$ that contains $w$. $w$ is guaranteed to match to a partition of
    the dictionary that has size $\le |D_j|/k$. Thus, round $j+1$ of the
    attack will search a dictionary of size at most $|D_j|/k$ that also
    contains $w$.  The algorithm repeats until the dictionary has size
    one. At this point, \textsc{RecoverQuery} returns the only word in the dictionary, $w$. Thus,
    it takes $\left\lceil \log_k{|D|} \right\rceil$ rounds to complete the
    attack, where $D$ is the initial dictionary.
\end{proof}

The attacker must also inject a significant number of files. We show that the
number of files, along with the file size, measured in number of words, is not
too large.

\begin{claim}
    For any dictionary  $D$ and for any word $w \in D$, let
    $q_w$ be a private query for $w$ and $k$ be the number of partitions. Then,
    the total number of files injected by \emph{\textsc{RecoverQuery}}$(q_w, k)$ is $k
    \left\lceil \log_k{|D|} \right\rceil$.
\end{claim}
\begin{proof}
    During a single round of the attack, the words in the $i$-th partition of
    the dictionary must be distributed among $i$ unique files, so that the
    number of results for the query $q_w$ during the next round will be increased by
    $i$ if $w$ was in that partition. The maximum value for $i$ is $k$, the
    number of partitions.  Therefore, each round requires injecting at least
    $k$ files. There are $\left\lceil \log_k{|D|} \right\rceil$ rounds according
    to Claim~\ref{claim:numrounds}, so we require a total of $k\left\lceil
    \log_k{|D|} \right\rceil$ file injections.
\end{proof}

\begin{claim}
    For any dictionary $D$ and for any word $w \in D$, let
    $q_w$ be a private query for $w$ and $k$ be the number of partitions. Then,
    the total number of words injected by \emph{\textsc{RecoverQuery}}$(q_w, k)$ is
    $O(k|D|)$.
\end{claim}
\begin{proof}
    A dictionary $D_j$ is searched during round $j$ of the attack. Each
    word in partition $i$ of the dictionary appears $i$ times during round $j$.
    Each partition has size $\frac{|D_j|}{k}$. Therefore, the total file size
    injected during this round is $\frac{|D_j|}{k} \left(1 + 2 + \dots + k\right)
    = O(k|D_j|)$.

    Each round reduces the size of the dictionary searched by a factor of $k$,
    so $|D_{j+1}| = \frac{|D_j|}{k}$. According to Claim~\ref{claim:numrounds},
    there are $\left\lceil\log_k{|D|}\right\rceil$ many rounds. Then, if the
    initial dictionary has size $|D|$, the total file size injected across all
    rounds is:
    \begin{align*}
        & k|D| + k\left(\frac{|D|}{k}\right) + k\left(\frac{|D|}{k^2}\right) + \dots + k\left(\frac{|D|}{k^{\lceil\log_k{|D|}\rceil}}\right) \\
        < \, & k|D| \left( 1 + \frac{1}{k} + \frac{1}{k^2} + \dots \right) 
        = \, O(k|D|)
    \end{align*}
\end{proof}

The attack presented can recover a user's query on a generic secondary index
with perfect accuracy, even when the file contents and metadata, except
result counts, are hidden. The number of results returned is indeed the only information we assume,
and we do not require knowledge of the distribution of the query dictionary. 

\begin{figure}[t]
\small
    \centering
    \begin{tabular}{c p{5.7cm}}
        \toprule
        \thead{Notation} & \thead[l]{Definition} \\\midrule
        $\mathcal{D}$ & The database, a secondary index mapping words to the
        files they are associated with. \\\midrule
        $q_w$ & A query for the word $w$, where $w$ is hidden. \\\midrule
        $D$ & The dictionary, a set of words probed by the attacker. \\\midrule
        $k$ & The number of partitions to search during each round. A higher
        $k$ means more files injected per round, but fewer rounds total. \\\midrule
        $n_j$ & $|q_w(\mathcal{D})|$, or the number of file results for the
        query on round $j$. For $j = 0$, this is the user's initial
        query dictionary. \\\midrule
        $m$ & A parameter for the single-round attack. A higher $m$ means more
        files injected, but higher expected accuracy. \\\midrule
        $s$ & A parameter for the noisy data attack. A higher $s$ means more
        files injected, but a greater possible amount of noise tolerated. \\
        \bottomrule
    \end{tabular}
     \caption{%
        Notation used in the described attacks.
    }
    \label{fig:notation}
\end{figure}

\begin{figure*}[ht]
\small
\renewcommand{\arraystretch}{1.2}
    \centering
    \begin{tabular}{l  c  c  c}
        \toprule
        \textbf{Attack type} & \textbf{Number of replays} & \textbf{Total files injected} & \textbf{Total words injected} \\\midrule
        Base attack & $\left\lceil \log_k{|D|} \right\rceil$ & $k\left\lceil \log_k{|D|} \right\rceil$ & $O(k|D|)$ \\%\hline
        Single-round attack & 1 & $m|D|$ & $O(m|D|^2)$ \\%\hline
        File padding (base-2 tiers) & $k\left\lceil \log_k{|D|} \right\rceil$ & $O\left( n_0 k |D|^{\log_k{2}} \right)$ & $O(n_0|D|)$ \\
        Noisy data & $\left\lceil \log_k{|D|} \right\rceil$ & $sk\left\lceil \log_k{|D|} \right\rceil$ & $O(sk|D|)$ \\\bottomrule
    \end{tabular}
    \caption{%
        The overheads of each type of attack, in terms of the
        number of query replays, files injected, and words injected. 
    }
    \label{fig:complexity}
\end{figure*}

\subsection{Single-round Attack}
\label{s:attack:noreplay}

\begin{algorithm}[t]
\small
    \caption{
        Pseudocode for the single-round attack.
        Input $m$ represents the tradeoff between file injection and
        accuracy. 
    }\label{alg:noreplay}
    \begin{algorithmic}[1]
        \Procedure{\textbf{SingleRoundInit}}{m}
        \State $\mathcal{D} \gets$ the initial database
        \State $D \gets$ keyword universe
        \For{$i$ in $[1, \dots, m \times |D|]$}
            \State $F_i \gets$ an empty file
        \EndFor
        \For{$index$ in $[1, \dots, |D|]$}
            \State $w \gets D[index]$
            \State Append $w$ to $m \times index$ unique files in $F$
        \EndFor
        \State $\mathcal{D} \gets \textsc{InjectFiles}(\mathcal{D}, F)$
        \EndProcedure
        \smallskip

        \Procedure{\textbf{SingleRoundRecoverQuery}}{$q_w$, $m$}
        \State $n \gets |q_w(\mathcal{D})|$
        \If {$n < m$}
        \State $w' \gets null$
        \Else
        \State $index \gets \left\lfloor \frac{n}{m} \right\rfloor$
        \State $w' \gets D[index]$
        \EndIf
        \State \Return $w'$
        \EndProcedure
    \end{algorithmic}
\end{algorithm}

The base attack relies heavily on the ability to replay the user's query.
Without the ability to replay, it is difficult to recover the
query with total accuracy without some knowledge of the 
distribution of words in the database. This is a
fundamental limitation of the attack---since we assume that the attacker cannot
read any file metadata other than the total number of results, the attacker cannot
differentiate between user-uploaded files and injected files during just a single
round of the user's query.

We now present a version of our attack that does not require the attacker to replay queries.
We show that the attacker can guess the user's query in a
single round with some degree of accuracy if he can inject a larger number of files. 
Moreover, if the universe of possible keywords is 
smaller in size, then the attacker can identify the query with high probability. For example, consider an attacker
who knows that Alice is sick, and wants to identify what disease she has by recovering her queries. The attacker can use the set of common diseases as the dictionary of possible keywords, the size of which is on the order of tens.
\iffull
We call this attacker strategy a {\em known query set}. 
\fi
Note that using this dictionary still does not require knowledge of Alice's query distribution.
In fact, our attack will also permit the attacker to identify that the query of the victim is not in his query set.

The key idea is as follows.
Because the attacker has only one
round to complete the attack, he must inject enough files \emph{before} the
user sends his query such that the attacker can still recover the query with
some accuracy. Although he cannot inject files multiple times as in the base attack,
he may still be able to inject a large enough quantity of files so as to filter out the noise
from user-uploaded files that match the query. 
And, as we
will see from the analysis, this can actually be done in such a way that
\emph{multiple} queries can be recovered without requiring the attacker to execute the attack repeatedly per query, 
in contrast to the base attack.

We describe the attack formally in \cref{alg:noreplay}.
First, the attacker initializes the attack using \textsc{SingleRoundInit}. 
The attacker starts with a dictionary $D$ of candidate words, and chooses a constant $m$. 
Before the user sends his query $q_w$, the attacker injects $m|D|$ files into the database $\mathcal{D}$, such that the $i$-th word
in the dictionary appears in $mi$ files. Thus, each word
appears a unique number of times and is spaced apart by at least $m$ files.

Then, when the user queries $q_w$, the attacker estimates $w$ using
\textsc{SingleRoundRecoverQuery} (\cref{alg:noreplay}). The attacker
first reads the result set size $n = |q_w(\mathcal{D})|$. If $n < m$, then $w$ is
not in the attacker's dictionary, and no more information can be gained for
this particular query. If $n \ge m$, then the attacker guesses $w'$, the $i$-th
word in the dictionary, where $i = \lfloor \frac{n}{m} \rfloor$.
With some probability, the attacker's guess is correct and $w' = w$.

The question, then, is how to choose $m$ such that we maximize the probability
that $w' = w$. Clearly, the larger $m$ is, the better, since a larger $m$ can
filter out more noise from the user's uploaded files. Given some underlying
distributions of word and query frequency, we can write the precise probability
of the attack's success. 

Formally, let $Q$ be a probability distribution over the universe of words
where $Q(w)$ is the probability that the user will query $q_w$. Let
$\mathcal{D}_0$ be the initial database, before any file injections. Then,
$|q_w(\mathcal{D}_0)|$ equals the number of user-uploaded files that would
have been returned for $w$.

\begin{claim}\label{th:single1}
    For any query $q_w$, and any $m \ge 1$, the probability that
    \textsc{SingleRoundRecoverQuery$(q_w, m)$} outputs an incorrect $w'$
    is: \[ \Pr(w' \neq w) = \sum_{w, |q_w(\mathcal{D}_0)| \ge m} Q(w) \]
\end{claim}
\begin{proof}
  Consider a user query $q_w$. If $|q_w(\mathcal{D}_0)| < m$, then there are
    two cases for the result set size, either $w$ is in the dictionary
    probed by the attacker or $w$ is not. If $w$ is not in the dictionary, then
    the current database $\mathcal{D}$ is unchanged, so the count 
    observed is still $|q_w(\mathcal{D}_0)|$. Since this is less than $m$, the
    attacker will not output a guess $w'$, so we can ignore this case.
    Otherwise, suppose that $w$ is the $i$-th word in the attacker's
    dictionary.  Then, the number of injected files for $w$ is $mi$. The
    attacker will then guess the word at index $\lfloor
    \frac{|q_w(\mathcal{D})|}{m} \rfloor = \lfloor \frac{|q_w(\mathcal{D}_0)| +
    mi}{m} \rfloor = i$, so $w' = w$.

  The remaining case is when $|q_w(\mathcal{D}_0)| \ge m$. Then, it is
    guaranteed that $w' \neq w$, whether or not $w$ is actually in the
    dictionary probed by the attacker.  If $w$ is the $i$-th word in the
    dictionary, then the attacker will guess $w'$ with a dictionary index
    greater than $i$. Otherwise, the attacker will incorrectly guess that the
    user queried for a word in the dictionary.  Thus, the probability that the
    attacker guesses an incorrect $w'$ is the probability that the user will
    query for a word $w$ such that $|q_w(\mathcal{D}_0)| \ge m$.  This is
    $\sum_{w, |q_w(\mathcal{D}_0)| \ge m} Q(w)$.
\end{proof}
This claim implies that if a large enough $m$ is chosen, then the attacker will
be able to perfectly recover all user queries. For instance, if $m$ is greater
than $\max_w{|q_w(\mathcal{D}_0)|}$, then the probability of an incorrect guess is 0.
Therefore, the better the attacker can estimate the distribution $|q_w(
\mathcal{D}_0)|$ and $Q$, the more he can
increase his probability of recovering the user's query correctly. Otherwise,
he will have to guess a large enough $m$ to ensure accurate query recovery, at
the cost of more file injection.

The query's success rate is also dependent on the dictionary of words chosen by
the attacker. If $Q(w) = 0$ for all $w \in D$, for example, the attacker will
not be able to output a correct guess. Ideally, the attacker would insert the
entire universe of words, but this is infeasible since the total number of
words injected is given by:
$m + 2m + \dots + |D|m = O(m|D|^2)$.
However, even if the attacker can only afford to
probe a small dictionary, he can still increase his chance of success
if he has some knowledge of $Q$; he can then choose to probe words that the
user is more likely to query.

There are two key advantages of this approach over the base scheme presented
in \cref{s:attack:basic}. First, the initial round of file injections can be
reused to recover multiple user queries over a long period of time. As long as
the attacker chooses a large enough $m$, the noise due to files that may be added
by the user later on can still be filtered out. The attacker can launch a
long-running attack in which he continuously probes for the same dictionary of
words by gradually increasing $m$ to match the rate at which real files are
added. Then, at any point in the future when the user queries for
a word in the dictionary, the attacker will be able to discover the word.

The second advantage is that this variation of the attack is mostly passive, in
that the attacker actively injects files once and then passively reads file
responses for the remaining duration. This is in contrast to the base attack,
in which the attacker must actively inject new files with every query response.
Thus, although the file injection overhead becomes higher and the success rate
is reduced, an attack without the ability to replay a user's query is still
both possible and practical.

\subsection{Attack Against File Padding}
\label{s:attack:padding}

An obvious countermeasure to the attack outlined in \cref{s:attack:basic}
is to use a cryptographic scheme that pads query responses to hide the number of files returned. 
Note that padding might not always be possible because it potentially adds 
nontrivial bandwidth overheads and hence increases costs for a system operator.

Padding interferes
with the attacker's ability to determine the number of files returned for a
query. However, as we show in this section, the attacker can still learn
some information. The
common issue in all of the following schemes is that the attacker retains the
ability to inject files. Thus, even if the attacker can no longer determine a
user's query, he can still inflate the
bandwidth overhead by injecting large enough files.

The simplest scheme would be to always pad to the worst case. Formally, the
largest possible response is $max_w |q_w(\mathcal{D})|$. 
The scheme must then pad every result set to this count, which is potentially very expensive. The attacker can
aggravate the problem by simply injecting a large number of files that all contain the same word, forcing the system to
send that many files in response to all requests. 
\iffull
Even without the attacker's injecting files, common words can appear in a large number of
user-uploaded documents.
\fi

A more practical scheme is to use \emph{tiered padding}. In this case, each
response is padded to one of several predefined sizes, or tiers. For example,
one could choose to use base-2 exponential padding, so that each response size
is rounded up to the nearest power of 2.

Tiered padding can deter the base attack, but comes at the cost
of expensive bandwidth overheads. Here, we analyze the number of files that the
attacker must inject and the bandwidth overhead for the server. We use
exponential tiered padding for the analysis, but a similar analysis applies
to any padding scheme. Recent works~\cite{Kellaris:2017:privsearch,
Kuzu:2014:privsearch} propose more efficient padding schemes that add
probabilistic noise to the result set size; in~\cref{s:attack:noise} we describe an
extension of our attack that applies to such scenarios.

Recall that on every round, the attacker records $n_{j-1} =
|q_w(\mathcal{D})|$, the number of files returned to the user's query at the
end of the previous round. Under this padding scheme, $n_{j-1} = 2^p$ for some
$p$. The actual number of files is then in the range $(\frac{n_{j-1}}{2}, n_j]$. In
the next round, the attacker must inject enough files to ensure that the query
will be padded to the next highest tier, so that there is a measurable
difference in the query's number of results. This is at least $\frac{n_{j-1}}{2}$
files.

In a direct translation of the base attack, the
attacker has to inject $k$ partitions of the dictionary in such a way that he
can differentiate between the partitions. Then, the attacker would have to
inject $\frac{n_{j-1}}{2}$ files for the first partition, $\frac{n_{j-1}}{2}
\times 2$ for the second, and so on. This leads to 
${n_{j-1}}{2^{k-2}}$ files injected for a single round. Even worse, $n_j =
n_{j-1} 2^{k-2}$, so the next round will also require an exponentially larger
number of files to be injected.

The number of files injected can be reduced by injecting the partitions one at
a time, with $\frac{n_{j-1}}{2}$ files per partition. After all of the
files for a partition are injected, the attacker can replay the query to
measure if the user's query matched that partition. This way, the attacker can
inject $\frac{kn_{j-1}}{2}$ files per round. However, this requires
increasing the number of query replays by a factor of $k$, since each round now
requires $k$ replays instead of one.

\begin{claim}\label{th:pad1}
    Consider a database that uses base-2 exponential tiers to pad query
    responses.  For any user query $q_w$, let $n_0 = |q_w(\mathcal{D})|$, the
    query response on the initial database. For any attacker dictionary
    $D$ and any number of partitions $k$ used during the search, the total
    number of file injections necessary to recover the query is $O(n_0 k
    |D|^{\log_{k}2})$.
\end{claim}
\begin{proof}
    During round $j$ of the attack, where $n_{j-1}$ is the observed number of
    files returned by $q_w$ during the previous round, the attacker must inject
    $\frac{n_{j-1}}{2}$ files for each of the $k$ partitions. Then, the
    attacker must inject $\frac{kn_{j-1}}{2}$ files during round $j$. Each
    round doubles the number of files returned, so that $n_j = 2n_{j-1}$.
    There are $\left\lceil \log_k{|D|} \right\rceil$ rounds by the same analysis
    as in Claim~\ref{claim:numrounds}. Then, the attacker must inject a total
    of $O(n_0 k2^{\log_{k}|D|}) = O(n_0 k |D|^{\log_{k}2})$ files to recover $q_w$.
\end{proof}

\begin{claim}\label{th:pad2}
    For any query $q_w$, the total size of files injected, measured in number
    of words, is $O(n_0|D|)$.
\end{claim}
\begin{proof}
    To compute the total number of words injected, we first consider the number of
    words injected during round $j$. The number of files injected is
    $\frac{kn_{j-1}}{2}$. Each file contains a copy of a single partition of
    the current dictionary.  Since the dictionary size is reduced by a factor
    of $k$ with each round, the current dictionary has size $\frac{|D|}{k^{j}}$.
    Then, the total number of words injected during a single round is
    $\frac{n_{j-1}|D|}{2k^{j-1}}$. $n_{j} = n_0 2^j$, so the number of words
    injected during a single round is
    $O\left({n_0|D|}\left(\frac{2}{k}\right)^{j-1}\right)$. With $k \ge 2$, this
    gives us a total of $O\left( n_0|D| \right)$ words injected across all
    rounds.
\end{proof}

While the overhead for the attacker is significant, this analysis does not take
into account the substantial bandwidth costs for the client. Every
query may require nearly doubling the query's number of results. To answer the
attacker's replayed queries, the cryptographic scheme needs to pad the files with
approximately as much data as the attacker must inject, to hide the
number of files returned. This can be prohibitively expensive for a database
system.

In the PIR setting, an even more effective version of our
attack is possible: the server has the ability to also delete data in addition
to injecting it. So the attacker can toggle the number of results for a keyword
over multiple padding sizes, instead of just increasing it.

\begin{figure*}[t]
\begin{minipage}[t]{0.28\textwidth}
\includegraphics[width=\linewidth]{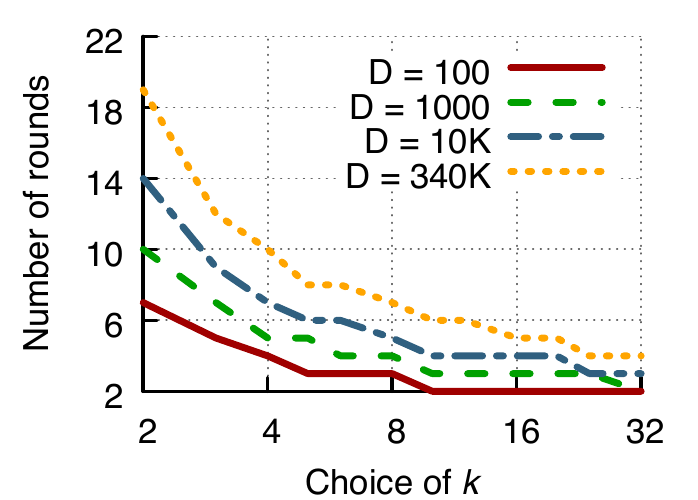}
\caption{\textbf{(Base attack)} Number of rounds required to identify a keyword with varying choices of
$k$ and different dictionary sizes.}
\label{fig:basic:rounds}
\end{minipage}
\hspace{0.21cm}
\begin{minipage}[t]{0.28\textwidth}
\includegraphics[width=\linewidth]{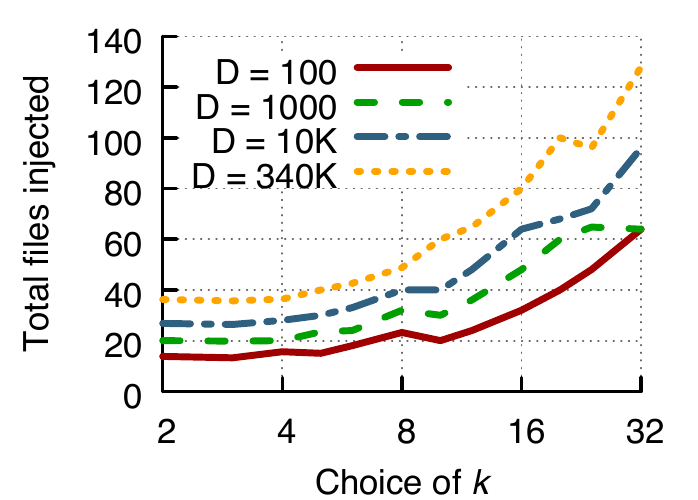}
\caption{\textbf{(Base attack)} Total number of files injected with varying choices of $k$, across different dictionary sizes.}
\label{fig:basic:files}
\end{minipage}
\hspace{0.21cm}
\begin{minipage}[t]{0.4\textwidth}
\includegraphics[width=\linewidth]{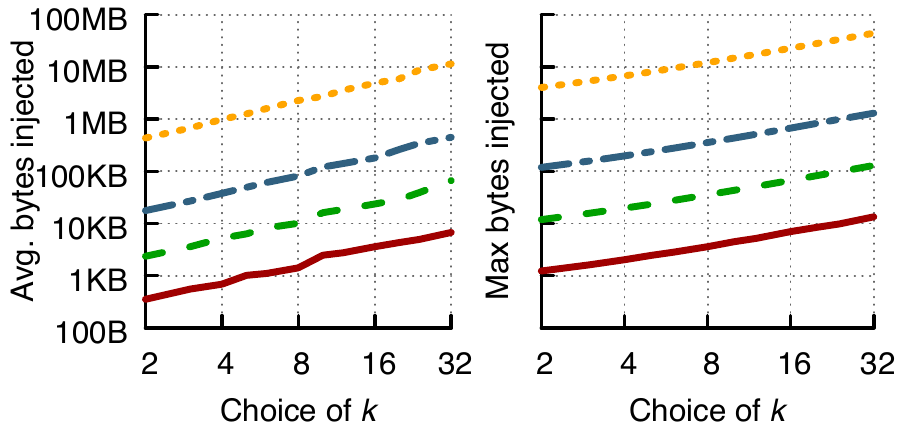}
\caption{\textbf{(Base attack)} Number of bytes injected for varying choices of
$k$, across different dictionary sizes: (left) average bytes per round; (right)
maximum bytes across rounds.}
\label{fig:basic:bytes}
\end{minipage}
\end{figure*}

\subsection{Attack with Noisy Data}
\label{s:attack:noise}

In \cref{s:attack:basic}, we assume that the attacker can identify precise
result counts. That is, we assume that the observed number of results is exactly
equal to $|q_w(\mathcal{D})|$, to the number of files that matched $w$ in the
database $\mathcal{D}$. This allows the attacker to precisely measure the
change in query result sets between rounds due to injected files.

However, the change measured may not exactly equal the number of injected
files. %
For example, if the cryptographic scheme adds some noise to
the result sets, then the attacker cannot precisely identify the change in
result counts due to file injection. Another example appears in some searchable encryption schemes, where the results are batched together in blocks (say $m$ results per block), or when using ORAMs~\cite{Roche:vORAM} that attempt to hide the number of results within an ORAM Path. 
For the first, the attacker observes the number of blocks
so it can estimate the actual number of results within an error of \txtapprox$m$. For the second, we discuss in \cref{sec:relwork} that such ORAMs still reveal an approximate number of results in some realistic settings. 

In this section, we show that our attack still has a significant chance of success 
even if there is some noise in the volume measurements. %
Formally, if the attacker expects a noise of at most $s$ files in some time
interval, this means that for each word $w$, the database system can add up to
$s$ elements to the database $\mathcal{D}$. Each of these elements is of the
form $(w, f)$, where $f$ is a dummy file. If a user queries $q_w$ on every
interval, he can expect an increase of at most $s$ files with each new query.

Similar to ideas presented in the above scenarios, the attacker can still
recover a query $q_w$ if he can inject more files to filter out noise in the
database. In particular, with an expected noise of $s$ in between rounds of an
attack, he can repeat a word in the $i$-th partition $si$ times instead of just
$i$. If the user's query word $w$ falls in partition $i$ during round $j$, then
the number of results observed will be $n_j \le n_{j-1} + si + s$. Then, to
determine the partition that the user's query belonged in, he can compute
$\lfloor \frac{n_j - n_{j-1}}{s} \rfloor = i$.

Thus, assuming that the attacker can correctly guess the maximum noise that
will be added during any round, the attacker can still recover the user's query
with perfect accuracy. The attacker can estimate $s$ with some knowledge of the
application. For instance, he can record the rate of incoming email for an
average user, for example. Furthermore, if the attacker underestimates $s$ and
guesses a wrong partition, he will quickly discover his error since the query's
result set size is unlikely to match any partition during the following round.

The overhead to overcome noise is quite low; a factor of $s$ to the number of
files and words injected. Thus, even if the attacker is unable to perfectly
measure the number of results and/or block network traffic, he can still recover the
user's query with near-perfect accuracy.

\subsection{Queries with multiple keywords}
Our attacks so far focus on queries with single keywords.
However, applications may allow clients to issue queries with multiple keywords as well.
One way in which applications may handle such queries is by expressing the queries as a conjunction of the different keywords. For such cases, we describe an extension to our attack in Appendix~\ref{s:conjunctions}. 
We present two attacks for conjunctive queries---the first optimizes the number of required replays while the second reduces the number of files injected.

Applications may alternatively express multi-keyword queries as disjunctions of the different keywords instead. We leave attacks on disjunctive queries to future work. 
\section{Evaluation}
\label{sec:eval}

\label{s:eval}

In this section, we evaluate the overheads and accuracy for the base attack (\cref{s:eval:basic}) and its extensions.
We simulate various application settings and degrees of attacker ability, including an
attacker that cannot replay the query (\cref{s:eval:noreplay}), and a storage
system with file padding (\cref{s:eval:padding}).
Second, we present a case
study on Gmail to evaluate the feasibility of the attacker's abilities assumed
in the base attack in a real-world application (\cref{s:eval:gmail}).
We demonstrate successful attacks on Gmail by simulating a server-side adversary, that complete within a few minutes across a variety of dictionary sizes.

\subsection{Setup}\label{s:eval:setup}

In all experiments, we use the entire corpus of emails from the Enron email
dataset~\cite{Enron} as the queried documents, consisting of
\txtapprox500K emails belonging to 151 users and \txtapprox2.5GB in size. We extracted
keywords from this dataset by first stemming the words~\cite{Porter}, and then removing 675 stopwords. We next
filtered out any words that contained non-alphabetic characters, or were $\ge20$
or $\le3$ characters long. This gave us a total of \txtapprox259K keywords.
In our experiments, we only used the top \txtapprox123K keywords (\ie
those that appeared in $>3$ documents) in order to remove noise from the
dataset.

Since an attacker's dictionary may contain words that do not exist in the
queried documents, we supplemented the Enron keywords with a corpus of English
words~\cite{English}. Preprocessing the English words in a similar manner yielded a total of \txtapprox257K keywords.
The union of both datasets resulted in a universe of
\txtapprox342K keywords.

\subsection{Base Attack}
\label{s:eval:basic}

Assuming that the queried word is in the initial dictionary chosen by the
attacker, the base attack achieves perfect query recovery, with strict
bounds on the overheads necessary in number of query replays and data injected
(as described in \cref{fig:complexity}). Our simulation of the attack in
\cref{fig:basic:rounds,fig:basic:files} confirms the theoretical guarantees.

In this experiment, we build the attacker's dictionary $D$ by randomly sampling
keywords from the keyword universe. We pick the keyword queried by the user at
random from $D$ in order to stress
test the effort required by the attacker---a keyword not in the
$D$ would be trivially detected at the end of a single round without
requiring further replays. We then report the number of rounds required to
guess the keyword with $100\%$ accuracy for different choices of $k$ in \cref{fig:basic:rounds}, and the total number of files injected across rounds in \cref{fig:basic:files}. Recall
from \cref{s:attack:basic} that any instance of the attack converges after
exactly $\left\lceil \log_k{|D|}\right\rceil$ replays and $k\left\lceil \log_k{|D|}
\right\rceil$ files injected, where $k$ is an integer chosen by the attacker.
Thus, with a dictionary of fixed size $|D|$, the parameter $k$ represents a
tradeoff between the number of query replays required vs. the number of file
injections required. The attacker's choice of $k$ then depends on the attacker's
ability to replay the query and the rate at which files can be injected for the
target application.

We explore this tradeoff with fixed-size dictionaries in
\cref{fig:basic:bytes}, which demonstrates how the average number of bytes
injected per round increases with $k$ (while the number of rounds decreases).
In the worst case where the dictionary comprises the entire keyword universe
and $k=24$, the bytes injected per round is still less than 10MB, demonstrating
the feasibility of the attack. We also show the maximum number of bytes
injected across any round in \cref{fig:basic:bytes}, equivalent to the number
of bytes injected during the first round.  The number of bytes that the
attacker can inject during a single round must be at least as large as this
number. We find that even in the worst case, this is approximately 50MB.

\paragraph{Takeaway} The attack can be mounted easily even when queries are
replayed at most once, \ie the attacker can recover the keyword in merely two
rounds without having to inject more than several tens of MBs of data. As an
example, Gmail limits the size of emails to a comfortable 25MB~\cite{gmailsize},
and the attacker need only send 3-4 emails to the victim's inbox in order to
identify the query.

\subsection{Single-round Attack}
\label{s:eval:noreplay}

\begin{figure*}[t]
\begin{minipage}[t]{0.32\textwidth}
\includegraphics[width=\linewidth]{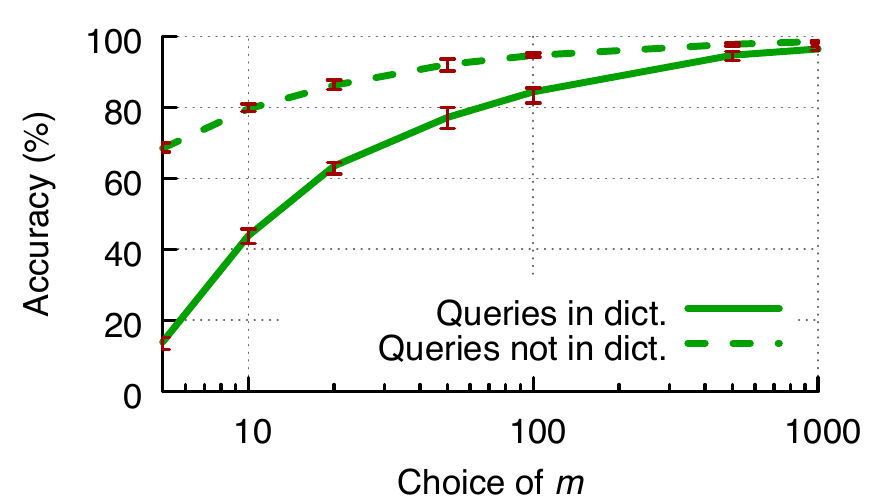}
\caption{\textbf{(Single-round attack)} Accuracy with varying $m$ across different dictionary sizes.}
\label{fig:noreplay}
\end{minipage}
\hfill
\begin{minipage}[t]{0.32\textwidth}
\includegraphics[width=\linewidth]{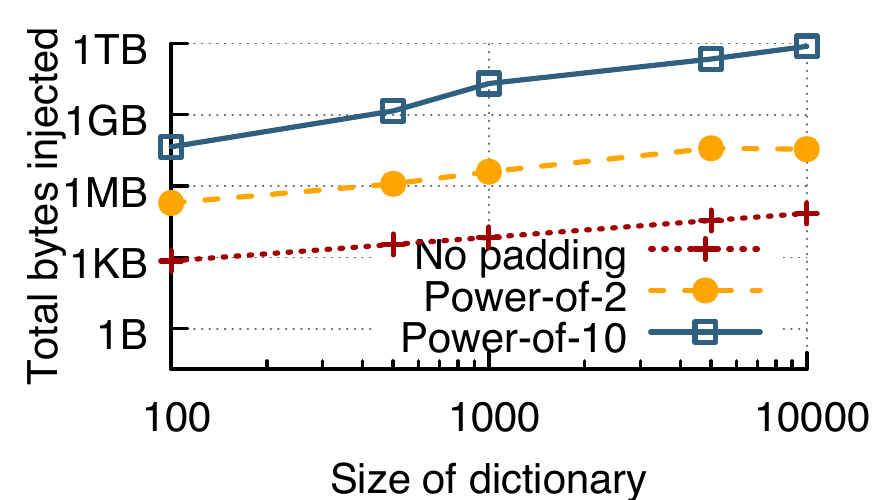}
\caption{\textbf{(File padding)} Overhead incurred by attacker 
when query responses are padded.}
\label{fig:pad}
\end{minipage}
\hfill
\begin{minipage}[t]{0.33\textwidth}
\includegraphics[width=\linewidth]{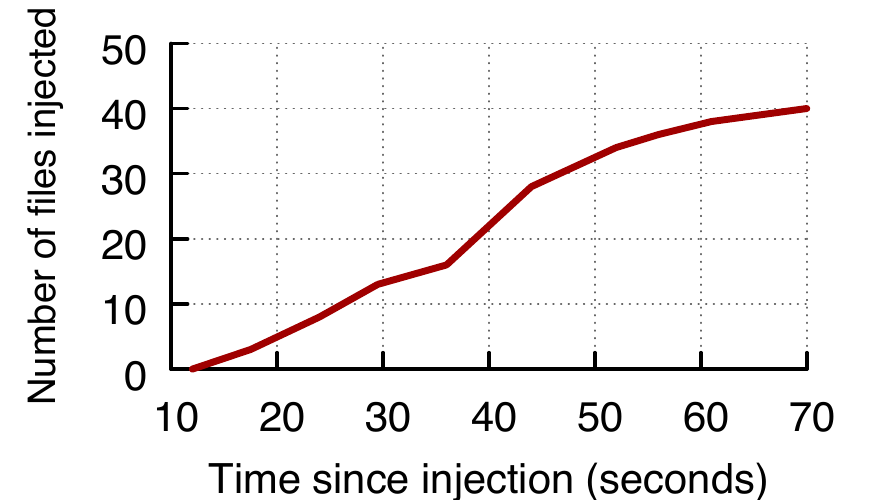}
\caption{\textbf{(Gmail)} CDF measuring the time it takes to inject files into the
    index for Gmail.
    }
\label{fig:gmail:injection}
\end{minipage}
\end{figure*}

In the single-round variation of the attack, we sacrifice some accuracy but do
not require the ability to replay the user's query. This variation of the
attack is also stronger in that it can be used to identify multiple queries
over a long period of time, whereas the base attack must be
instantiated once for every query of interest. Also, this variation is a mostly
passive attack, since the bulk of the attack is spent reading query responses,
rather than also injecting files in an online fashion.

We evaluate this attack by measuring its accuracy while varying the parameter $m$ chosen by the attacker. Recall that $m$ represents the tradeoff between the file injection overhead
and the number of files injected (\cref{s:attack:noreplay}). In this
experiment, we only query keywords that exist in the Enron dataset, since
keywords that do not exist in the dataset will always be accurately detected for
any choice of $m$.

For each value of $m$, we measure the attack's accuracy in two scenarios:
(i)~when the queried keyword is in the attacker's dictionary and the attacker
guesses the keyword; and (ii)~the keyword is not in the dictionary,
and the attacker determines that the keyword is not of interest. In each
scenario, we first fix the dictionary, and then inject a single round
of files at the beginning of the simulation for a chosen value of $m$. We then
query 1000 randomly selected keywords and measure the percentage of accurate
guesses.
\cref{fig:noreplay} plots the accuracy of guesses as $m$ increases, averaged
over different dictionary sizes.

Predictably, the accuracy of the attack improves with $m$, while the number of
file injections required also increases. For $m = 5$, we need only inject $5|D|$
files, but achieve an
accuracy rate of only \txtapprox$14\%$ for words that belong to the attacker's
dictionary. For $m = 1000$, we achieve an accuracy of
\txtapprox$96\%$, but must inject $1000|D|$ files.
In practice, a lower value of $m$ might not only suffice but also be
feasible: for $m=20$ we achieve an accuracy of \txtapprox$64\%$, while the number
of bytes injected for a dictionary of size 10K is \txtapprox$7$GB.

\paragraph{Takeaway} The barrier to mount an attack is higher in the absence of
replays, and the attacker needs to inject several GBs of data to identify a
keyword with reasonable confidence. However, in scenarios where the attacker's
dictionary contains a small number of words (when the attack knows a candidate list of queries, as discussed in \cref{s:attack:noreplay}), the feasibility of the attack
increases proportionately. For example, an
attacker who wishes to identify the disease that a victim might have only needs 
a dictionary of \txtapprox$950$ keywords~\cite{cdcdiseases}. If the
attacker is interested
only in sexually transmitted infections (STIs), then the size of the dictionary
drops to \txtapprox$27$ keywords~\cite{STIs}, increasing the feasibility of the
attack manifold.

\subsection{Attack Against File Padding}
\label{s:eval:padding}

We assess the feasibility of the attack when the server pads query responses
to one of several predefined sizes. In such cases, though it is still possible
for the attacker to guess the queried keyword,
it also requires greater effort. 
As described in \cref{s:attack:padding}, the attacker can choose to either
minimize the number of rounds, or minimize the number of files injected. In this
experiment, we evaluate the latter strategy. Specifically, we measure the
overhead incurred by the attacker when the server pads responses to powers of
2 and 10, and compare it with a baseline where the responses are unpadded.

We build the attacker's dictionary by randomly selecting keywords from the Enron
dataset, and then measure the overhead for {\em each} keyword in the dictionary.
We use this setup to stress the number of files the attacker will have
to inject, since responses for non-existent keywords will get padded to a size
of 1 by the server. 

\cref{fig:pad} illustrates the attacker's overhead in
terms of the total number of bytes injected with varying dictionary sizes.
When responses are padded to a power of 2, the attacker has to inject a feasible
\txtapprox$35$MB on average to mount the attack, with a large dictionary
size of 10K keywords. 
\iffull
However, the overhead increases drastically for the same
dictionary size when the responses are padded to a power of 10, requiring
\txtapprox0.75TB to be injected. Even so, a small dictionary size of 100 keywords
only requires \txtapprox45MB of injected data to pinpoint the
keyword.
\else
Though the overhead increases dramatically when the responses are padded to a power of 10, 
a small dictionary size of 100 keywords still requires only \txtapprox45MB of injected data to pinpoint the
keyword.
\fi

\paragraph{Takeaway} Padding responses is a viable defense {\em only if}
the following conditions hold simultaneously: (i)~the quantum of
padding is high, and (ii)~the attacker's dictionary of interest is large. In
all other situations, the attack remains feasible as demonstrated above.

\subsection{Case Study: Gmail Inbox Search} %
\label{s:eval:gmail}

\begin{figure}[tp]
\small
\begin{center}
    \begin{tabular}{c  c  c  c  c  c}
        \toprule
        \multirow{2}{*}{{$|D|$}} & \multirow{2}{*}{{$k$}} & \multicolumn{2}{c}{\thead{No.\ of replays}} & \multirow{2}{1.1cm}{\centering\textbf{Total injected emails}} & \multirow{2}{1.2cm}{\centering\textbf{Attack duration}}\\
        \cline{3-4}
        & & \thead{Theoretical} & \thead{Actual} & & \\
        \midrule
        10 & 10 & 1 & 1 & 10 & 1min \\ 
        100 & 10 & 2 & 2 & 20 & 2min 5s\\
        1K & 32 & 2 & 2 & 63-64 & 2min 5s\\
        10K & 22 & 3 & 3 & 64 & 5min 6s\\
        100K & 18 & 4 & 5 & 71 & 7min 10s\\\bottomrule
    \end{tabular}
    \caption{Attack parameters and duration for Gmail across various dictionary sizes.
    }
    \label{fig:gmail:diffdicsizes}
    \end{center}
\end{figure}

So far, we have experimentally validated the theoretical performance of our attacks across various parameter choices. 
We now demonstrate the practical feasibility of our attack in real-world applications by attacking Gmail's inbox search feature.
We attack the real Gmail web client, assuming that the server maintains a secondary index over the user's 
(encrypted) emails and only learns the volume of query results. 
Note that this is not currently the case, and Gmail's servers have access to far more leakage than simply the volume of results.
However, our aim is to demonstrate that even if Gmail (and other real-world applications) were to deploy sophisticated privacy-preserving mechanisms at the server such as ORAM or PIR, the volume of query results remains a potent source of leakage.
Therefore, in this experiment, we attack the real Gmail web client by simulating such a server-side adversary using a man-in-the-middle proxy; we simulate the adversary because we don't have actual control over Gmail servers.

We first show that the attacker can indeed meet the two key requirements of file injection and automatic query replay.
Subsequently, we perform an exhaustive experiment across a wide range of dictionary sizes (10 to 100K) to determine the {\em minimum amount of time} required to mount a successful attack on Gmail. The parameters of our attack are governed by the following constraints: (i)~the periodicity of replays in the Gmail client; (ii)~the time it takes to inject files into a user's inbox; and (iii)~the pagination limit in Gmail (which upper bounds the total number of injections).
We find that for a small dictionary of size 10, a successful attack can be mounted within 1 minute from start to end; for a large dictionary with 100K words, an attack completes successfully in around 7 minutes (see \cref{fig:gmail:diffdicsizes}).

We now describe our methodology in more detail.

\paragraph{Setup}
Since we don't have control over Gmail servers, we simulate a server-side adversary using a man-in-the-middle (MITM) HTTPS proxy~\cite{mitmproxy}. Specifically, we launch the Gmail web client on a browser within a guest virtual machine, and launch the MITM proxy on the host. We reroute all host network traffic through the MITM proxy. Subsequently, we install the proxy's certificate at the client browser in order to simulate a server-side adversary.
At this point, all TLS network traffic to and from the client browser passes through the MITM proxy, which it can then examine and manipulate.

\paragraph{Query replay}
Once a user issues a query, we use the MITM proxy to stimulate automatic query replay by simply dropping the HTTP responses returned by the server. After a period of time, the Gmail web client retries the query automatically, without user intervention. Specifically, we find that the client replays the query every 1--3 minutes in the absence of a response.
To the user, it simply appears as if the client has a bad network connection.

\paragraph{File injection}
File injection in Gmail is simple; the attacker requires a separate Gmail account to send emails to the victim.
For the base attack, the attacker must send $k$ emails in each round and also be sure that they are all indexed by the next replay (\ie at least 60s).
We determined the rate at which emails could be injected (\cref{fig:gmail:injection}) to show that it is feasible to index a sufficient number of emails.
We found that after injecting 40 emails of size 10KB each, 36 were visible in the user's mailbox 60 seconds later, shown in \cref{fig:gmail:injection}.
Thus, within a time window of 60s, the attacker can pick any value less than 36 as a safe option for $k$.

\paragraph{Volume leakage}%
In this experiment we assume that the proxy directly obtains the exact result set size from the server, since we simulate a server-side adversary. However, we find that Gmail has a maximum pagination limit of 100, \ie the server returns at most 100 results in response to a query. The pagination limit constrains the parameter regime of our attack, in that it upper bounds the total number of files that can be injected by the attacker over the duration of the attack.

\paragraph{End-to-end attack}
The aim of the experiment is to minimize the time it takes to launch a successful attack.
However, the constraints discussed above---the periodicity of replays, the time it takes to inject files, and the pagination limit---restrict the parameter regime within which an attacker can operate. Therefore, we start by computing the optimal parameters required for mounting a successful attack within the space of possible parameters. Next, we attack Gmail using the computed parameters and report the end-to-end duration of the attack.

 Since the Gmail application has a fixed periodicity of replays, the attack duration is directly governed by the number of replays required for the attack. Hence, given a dictionary size $|D|$, our aim is to minimize $\log_k |D|$, where $k$ refers to the number of files that need to be injected per round. However, given the pagination limit of $\ell = 100$, we require that the total number of injected files $k \times \log_k|D|$ be less than $\ell$. At the same time, $k$ should be less than 36 given the time it takes to inject files.

We therefore solve the following optimization problem:
\begin{align*}
\text{minimize}\quad & \log_k|D| \\
\text{subject to}\quad & k \times \log_k|D| < 100\\
\text{and}\quad & k < 36
\end{align*}
\cref{fig:gmail:diffdicsizes} summarizes our findings for varying sizes of the attacker's dictionary, 10 to 100K. Note that the total number of injected emails is sometimes marginally less than $k\times \log_k|D|$. This is because $\log_k|D|$ is not always an integer, and therefore files of interest across subsequent rounds may sometimes contain less than $k$ keywords. Additionally, for $|D| = 100K$, our attack requires an extra round of replay because the size of the injected files in the first round were large, increasing the time it took for files to get sent and indexed.

Overall, our experiment demonstrates the feasibility of volume-based attacks on Gmail, which can be successfully completed within a matter of minutes depending on the size of the attacker's dictionary.
In addition, the attack is difficult to detect because during the course of the attack, the user only sees a suspended connection. The user only makes a single query, and the Gmail client automatically replays the query in the background. During this time, emails injected by the attacker are not delivered to the user's web client, and only modify the server-side index. The user may later see the injected emails, but only after the attack successfully completes.

\section{Mitigations}
\label{sec:mitigate}

In this section, we discuss mitigations of our attack.
Overall, we believe it is difficult to eradicate our attack in all settings.
Injection is often fundamental to application functionality, worst-case padding is
too expensive, and replay could be a legitimate user or application action, as
discussed in~\cref{s:abilities}. Nevertheless, based on our evaluation in
\cref{s:eval}, we believe that the mitigations proposed below could significantly
reduce the extent of the attack by limiting the attacker's abilities (\cref{s:abilities}) or making the attack too expensive to mount.
However, with enough resources, it is possible for the
attacker to defeat some of these through the attack
extensions described in \cref{s:attack}.

\paragraph{Preventing volume leakage} 
Strategies that reduce the attacker's ability to measure the
number of files contained in a response can be effective in preventing volume-based attacks. 
As discussed in \cref{s:attack:padding,s:eval:padding}, padding query results to an upper bound 
might help mitigate the attack by increasing the attacker's burden, and can be effective in hiding the response size.
At the same time, it results in unaffordable overheads for many applications~\cite{Naveed15}, as also demonstrated by our analysis in \cref{s:eval:padding}.
We believe that to varying degrees, this is a property of all
padding schemes.

A more practical way to hinder the attacker is to inject some noise in the
responses. This requires little overhead in server-client bandwidth
compared to the attacker's overhead: an additive factor of $k$ per
query compared to a multiplicative factor of $k$ per attack. This
countermeasure is also simple to implement: add a random number of
dummy files to every response and have the client filter them out.
Note that while this increases the attacker's overhead, it does not wholly preclude the
attack, as we described in \cref{s:attack:noise}.

Another method is to limit the number of results that can be fetched at a time.
The user must explicitly request further results if needed.
While stricter limits on the number of results lowers the feasible dictionary size for the attack (thereby increasing the attacker's burden), it might also have an adverse impact on user experience.

\paragraph{Preventing file injection} File injection is arguably the most difficult to
defend against, since it is often a part of the target application's
functionality. For example, an email inbox search feature is not much use if it
can only search for keywords within emails that were sent by the user, and not to the user.

Thus, we believe that the main defense here is rate-limiting and detection. In
the email application (\cref{s:instances:email}), this would require the server to actively
filter out suspicious emails. 
As we found in
\cref{s:eval:gmail}, applications such as Gmail already rate-limit emails;
however, this was not enough to defeat the
attack. %

\paragraph{Preventing query replay} The most effective way to prevent the base attack is
to block query replays. Query replays are a feature of applications such as
Gmail that produce the illusion of a seamless connection during limited network
connectivity (\cref{s:eval:gmail}). A possible countermeasure is to include a
unique query ID for each request, so that the server can detect and filter out
duplicate requests.

The main disadvantage of such an approach is that the server would then have to
record and replay past responses in order to both prevent the attack and keep
the application available. Long-running user sessions would have to be
garbage-collected, potentially sacrificing correctness. More crucially,
web servers are often replicated for performance and fault tolerance. Ensuring
consistency for duplicate queries in such settings is well-known to be
expensive, if even possible~\cite{gilbert2002brewer}. Finally, this
countermeasure does not prevent against the single-round attack described in
\cref{s:attack:noreplay}.

\section{Related Work} %
\label{sec:relwork}

To  compute on encrypted data, the community has developed a rich set
of cryptographic schemes and protocols, as well as
encrypted database systems.  
A recent set of attack papers study
the information an attacker can obtain from these schemes and systems, termed {\em leakage-abuse} attacks by Cash \etal~\cite{CashJJKRS13}.  
Many attacks in this category leverage leakage from data relations or access
patterns, and few works target oblivious schemes and systems relying only on volume leakage, as our work does.

We now briefly discuss cryptographic schemes and systems that leak result volumes, followed by related attacks on these systems.

\paragraph{Cryptographic schemes and systems}\label{s:r-schemes}\label{s:r-systems}
There are a multitude of ways to access or compute on encrypted data,
such as property-preserving/ property-revealing encryption~\cite{boldyreva:ope, boldyreva:ope-revisited,popa:mope,KS14:optimalAvgOPE,WuLewiRange} or searchable encryption~\cite{SongWP00, CurtmolaGKO06, KamaraPR12, CashJJKRS13, OgataKKM13, CashJJJKRS14, LauCSJLB14, Kurosawa14, NaveedPG14, StefanovPS14, HeAJSS14, Bost16}. 
For a comprehensive survey, see~\cite{Smart:survey}.
Here we focus on systems leveraging ORAM or PIR that leak the volume of results (and are thus vulnerable to our attacks).

ORAM techniques~\cite{GoldreichO96,StefanovDSFRYD13} and PIR schemes~\cite{PIR:survey} enable a client to access data items stored at the server without the server knowing the query requested. These two types of schemes consider different models and employ different techniques, but ultimately, the goal of both is to hide the contents of the query from the server. 

Many works leverage ORAM for different purposes. For example, ObliviStore~\cite{Stefanov:2013} and CURIOUS~\cite{Bindschaedler:2015} show how to use ORAM for cloud storage. TaoStore~\cite{SahinZALT16} shows how to support asynchronicity in multi-user cases. 
These systems leak the volume of results to the server.
Oblix~\cite{mishra:oblix} builds a search index over ORAM and pads or truncates the set of results to a fixed size.
Roche \etal~\cite{Roche:vORAM} propose an ORAM scheme (called vORAM) that supports variable-sized data blocks by including them within an ORAM node (or bucket) on the same path, but  our attack with noisy data in \cref{s:attack:noise} can still work on these schemes. While such a scheme hides the result volume to some extent, it limits the amount of data that can be included on a path in this way (say, $L$ files). Since the attacker can see how many ORAM paths are fetched on a query, he can estimate the number of results with an error margin of $L$. In the database setting, this error margin can be made relatively small because the database fetches the rows that match the keyword (not just the row identifiers), and these cannot all be stored on the same path.
Moreover, Naveed~\cite{Naveed15} demonstrates that, in general, extending ORAM schemes to hide the volume of results is (for a large fraction of queries) actually slower than streaming the database through the client.

Some works~\cite{SQLQuery:PIR:Goldberg, PIRPublicData, PIRBook} build SQL databases or keyword indices on top of PIR. For example, to perform an index search for a keyword $k$, the client performs PIR retrievals to traverse the index and select every value in the index. The server does not know which data items were fetched, but it still sees the number of results.

\paragraph{Related attacks}\label{s:r-attacks}
When considering the amount of leakage attacks exploit, there are at least three categories: 
attacks exploiting data relations, 
attacks exploiting  access patterns, and  
 attacks exploiting result set size but not access patterns or data relations. The last category is the most challenging because the attacker needs to work with the least amount of information. At the same time, this category is also the least studied. Our attack is in this last category, and we now discuss other volume-based attacks.

\iffull
\paragraph{Attacks exploiting data relations}
These attacks~\cite{DDC16:whatElseORE, GSBNR16:leakageAbuseORE, NKW15:inferenceAttacks} target property-preserving or property-revealing encryption. They exploit information revealed by the encryption schemes such as the order or frequency of data items. 
In comparison, our attack applies to stronger schemes (\ie semantically secure
and oblivious) that do not have such data relations leakage.

\paragraph{Attacks exploiting access patterns}
A set of works showed how to exploit the leakage from query access patterns in searchable encryption~\cite{IslamKK12, CashGPR15, ZhangKP16, LiuZWT14, AbdelraheemAG17, GrubbsMNRS16, PouliotW16, GiaruadABL17,STRESS,KennyAccessPattern}. 
Other works exploit access patterns for range queries or other relational
queries~\cite{Islam:2014:rangepatterns,KellarisKNO16,Dautrich:2013}. These
attacks rely crucially on the attacker knowing what is the set of documents returned for a query, and in some cases more (\eg if the same query is repeated). 
In contrast, our work applies to stronger schemes that do not reveal the access patterns or leak the document set.

\paragraph{Attacks exploiting result set size}

There is very little work done on understanding the leakage from result set size in oblivious settings, which is the space of our work.
\fi

Cash \etal~\cite{CashJJJKRS14} point out that if an attacker knows the exact number of times a keyword appears in a victim's documents, and if that result size is unique to this keyword, the attacker can identify the keyword when seeing the result size.
 In comparison, our attack does not assume the attacker knows the frequency of each keyword in a victim's index---indeed, when attacking a specific user in the email application, the attacker often does not have access to the victim's mailbox and does not know these counts. Moreover, many keywords don't have unique counts (\eg $99\%$ words in the Enron dataset, \cref{sec:eval}), in which case the attack of Cash \etal does not work.  %
Our attacks do not suffer from this limitation.

In seminal work, Kellaris \etal~\cite{KellarisKNO16} showed how an attacker can reconstruct the contents of a field in the database given only the volume of results.
However, their attack relies on a set of assumptions that are arguably not realistic in practice.
In particular, Kellaris \etal assume that 
(1)~the user makes range queries that are {\em uniformly} distributed on that column, a property on which their algorithm relies crucially; and (2)~the user makes $O(N^4 \log N)$ queries where $N$ is the size of the domain. Such a large number of queries is infeasible for the attacker to observe in many settings. 
Grubbs \etal~\cite{Grubbs18:Volume} improved upon the results of Kellaris ~\etal by demonstrating attacks that do not make assumptions on the distribution of queries, as long as \emph{all} possible $O(N^2)$ range queries are issued. As a result, for queries drawn from a uniform distribution, their attack requires $O(N^2\log N)$ queries to be issued. 
In recent work, Gui \etal~\cite{Gui19:Volume} further improved upon the result of Grubbs \etal by demonstrating attacks that require an order of magnitude fewer queries. However, their attack still assumes that the adversary is able to observe all possible queries that produce a bounded number of results, and that the database is dense (\ie all possible values occur in the database).

Our attacks differ from the attacks described above in assumptions as well as target.
In contrast to existing attacks, our attack requires only a single query to be issued by the user, followed by $O(\log |D|)$ replays (which, concretely, is often less than $10$ in number; \cref{s:eval}).
Our attack also makes no assumptions about the query distribution.
On the other hand, unlike the aforementioned works, our attack requires the ability to inject and sometimes replay queries, though we demonstrate realistic scenarios in which this can be achieved (\cref{sec:setup:application}).
Another difference is that the aforementioned attacks reconstruct the database using range queries, but not individual query keywords; we  reconstruct queries, but do not target the overall database. However, we note that reconstruction follows as a direct consequence of our attack, where the original counts for each keyword could be determined if queries for all possible keywords are issued.

In concurrent work, Blackstone \etal~\cite{Blackstone19} also propose a suite of volume-based attacks, some of which passively analyze the volume of query results based on some known data (similar to prior work), while two additional attacks rely on injecting files into the database, similar to ours.
In particular, their file injection attacks are conceptually similar to our base attack in \cref{s:attack:basic}.
The primary difference is that our attacks leverage query replay---we study the real behavior of many applications, and crucially, we find automatic query replay to be a common property; by leveraging this property, we are able to substantially improve the efficiency of our attacks.
The attacks of Blackstone \etal do not require queries to be replayed.
As a consequence, however, their attacks rely on an alternate set of assumptions.
First, they require the adversary to know the baseline volumes for all keywords in the dictionary, before the attack can be launched.
Second, for correctness, their binary search based attack requires the targeted keyword to have a unique volume in the baseline volumes.
Our attacks do not have such requirements, and our algorithm allows us to prune the search space faster, drastically decreasing the overall duration of the attack.
We also describe multiple extensions to the attack, including settings where the results are padded (\cref{s:attack:padding}) or noisy (\cref{s:attack:noise}).

\section{Conclusion}
\label{sec:conclusion}

We demonstrated a generic attack on encrypted databases that only leverages result size leakage. We showed that our attack can reconstruct queries
in a range of realistic settings, weakening the security  of these systems. Our attack is resistant to various mitigation strategies, and can reconstruct sensitive information even in situations where the result volumes are padded, the volume measurement is noisy, or the client application lacks the ability to replay queries. We showed the effectiveness of our attack via both theoretical bounds and an empirical evaluation, including a demonstration on the Gmail web application.

\ifcamera
{
\small
\section*{Acknowledgments}
We thank the anonymous reviewers for their helpful feedback. 
This work was supported by the NSF CISE Expeditions Award CCF-1730628, as well as gifts from the Sloan Foundation, Bakar Program,  Alibaba, Amazon Web Services, Ant Financial, Capital One, Ericsson, Facebook, Futurewei, Google, Intel, Microsoft, Nvidia, Scotiabank, Splunk, and VMware.
}
\fi

{
\bibliographystyle{IEEEtran}
\bibliography{IEEEabrv,bib/references,bib/str,bib/conf}

\begin{thebibliography}{10}
\providecommand{\url}[1]{#1}
\csname url@samestyle\endcsname
\providecommand{\newblock}{\relax}
\providecommand{\bibinfo}[2]{#2}
\providecommand{\BIBentrySTDinterwordspacing}{\spaceskip=0pt\relax}
\providecommand{\BIBentryALTinterwordstretchfactor}{4}
\providecommand{\BIBentryALTinterwordspacing}{\spaceskip=\fontdimen2\font plus
\BIBentryALTinterwordstretchfactor\fontdimen3\font minus
  \fontdimen4\font\relax}
\providecommand{\BIBforeignlanguage}[2]{{%
\expandafter\ifx\csname l@#1\endcsname\relax
\typeout{** WARNING: IEEEtran.bst: No hyphenation pattern has been}%
\typeout{** loaded for the language `#1'. Using the pattern for}%
\typeout{** the default language instead.}%
\else
\language=\csname l@#1\endcsname
\fi
#2}}
\providecommand{\BIBdecl}{\relax}
\BIBdecl

\bibitem{Fuller:EDBs}
B.~Fuller, M.~Varia, A.~Yerukhimovich, E.~Shen, A.~Hamlin, V.~Gadepally,
  R.~Shay, J.~D. Mitchell, and R.~K. Cunningham, ``{SoK: Cryptographically
  Protected Database Search},'' in \emph{Proceedings of the IEEE Symposium on
  Security and Privacy ({IEEE S\&P})}, 2017.

\bibitem{CashJJJKRS14}
D.~Cash, J.~Jaeger, S.~Jarecki, C.~Jutla, H.~Krawczyk, M.~Rosu, and M.~Steiner,
  ``{Dynamic Searchable Encryption in Very-Large Databases: Data Structures and
  Implementation},'' in \emph{Proceedings of the Network and Distributed System
  Security Symposium ({NDSS})}, 2014.

\bibitem{FJKNRS15:RichQueries}
S.~Faber, S.~Jarecki, H.~Krawczyk, Q.~Nguyen, M.~Rosu, and M.~Steiner, ``{Rich
  Queries on Encrypted Data: Beyond Exact Matches},'' in \emph{Proceedings of
  the European Symposium on Research in Computer Security ({ESORICS})}, 2015.

\bibitem{popa:cryptdb}
R.~A. Popa, C.~M.~S. Redfield, N.~Zeldovich, and H.~Balakrishnan, ``{CryptDB:
  Protecting Confidentiality with Encrypted Query Processing},'' in
  \emph{Proceedings of the ACM Symposium on Operating Systems Principles
  ({SOSP})}, 2011.

\bibitem{AEKKRV13:cipherbase2}
A.~Arasu, K.~Eguro, R.~Kaushik, D.~Kossmann, R.~Ramamurthy, and R.~Venkatesan,
  ``{A secure coprocessor for database applications},'' in \emph{Proceedings of
  the International Conference on Field Programmable Logic and Applications
  ({FPL})}, 2013.

\bibitem{TKMZ13:monomi}
S.~Tu, M.~F. Kaashoek, S.~Madden, and N.~Zeldovich, ``{Processing Analytical
  Queries over Encrypted Data},'' in \emph{Proceedings of the International
  Conference on Very Large Data Bases ({VLDB})}, 2013.

\bibitem{website:MSalwaysEncryptedDB}
{\relax {Microsoft SQL Server: Always Encrypted Database Engine}}.
  \url{https://msdn.microsoft.com/en-us/library/mt163865.aspx}.

\bibitem{website:skyhigh}
\relax{Cloud Threat Intelligence, Skyhigh Cloud Security labs, {Skyhigh}
  {Networks}}. \url{https://www.skyhighnetworks.com/}.

\bibitem{website:ciphercloud}
{CipherCloud: Cloud Data Protection Solution}.
  \url{http://www.ciphercloud.com}.

\bibitem{website:iqrypt}
{{iQrypt}: Encrypt and query your database}. \url{http://iqrypt.com/}.

\bibitem{poddar:arx}
R.~Poddar, T.~Boelter, and R.~A. Popa, ``{Arx: An Encrypted Database using
  Semantically Secure Encryption},'' in \emph{Proceedings of the International
  Conference on Very Large Data Bases ({VLDB})}, 2019.

\bibitem{boldyreva:ope}
A.~Boldyreva, N.~Chenette, Y.~Lee, and A.~O'Neill, ``{Order-Preserving
  Symmetric Encryption},'' in \emph{Proceedings of the Annual International
  Conference on the Theory and Applications of Cryptographic Techniques
  ({EUROCRYPT})}, 2009.

\bibitem{boldyreva:ope-revisited}
A.~Boldyreva, N.~Chenette, and A.~O'Neill, ``{Order-Preserving Encryption
  Revisited: Improved Security Analysis and Alternative Solutions},'' in
  \emph{Proceedings of the International Cryptology Conference ({CRYPTO})},
  2011.

\bibitem{popa:mope}
R.~A. Popa, F.~H. Li, and N.~Zeldovich, ``{An Ideal-Security Protocol for
  Order-Preserving Encoding},'' in \emph{Proceedings of the IEEE Symposium on
  Security and Privacy ({IEEE S\&P})}, 2013.

\bibitem{KS14:optimalAvgOPE}
F.~Kerschbaum and A.~Schr{\"{o}}pfer, ``{Optimal Average-Complexity
  Ideal-Security Order-Preserving Encryption},'' in \emph{Proceedings of the
  ACM Conference on Computer and Communications Security ({CCS})}, 2014.

\bibitem{WuLewiRange}
K.~Lewi and D.~J. Wu, ``{Order-Revealing Encryption: New Constructions,
  Applications, and Lower Bounds},'' in \emph{Proceedings of the ACM Conference
  on Computer and Communications Security ({CCS})}, 2016.

\bibitem{SongWP00}
D.~X. Song, D.~Wagner, and A.~Perrig, ``{Practical Techniques for Searches on
  Encrypted Data},'' in \emph{Proceedings of the IEEE Symposium on Security and
  Privacy ({IEEE S\&P})}, 2000.

\bibitem{CurtmolaGKO06}
R.~Curtmola, J.~Garay, S.~Kamara, and R.~Ostrovsky, ``{Searchable symmetric
  encryption: improved definitions and efficient constructions},'' in
  \emph{Proceedings of the ACM Conference on Computer and Communications
  Security ({CCS})}, 2006.

\bibitem{KamaraPR12}
S.~Kamara, C.~Papamanthou, and T.~Roeder, ``{Dynamic Searchable Symmetric
  Encryption},'' in \emph{Proceedings of the ACM Conference on Computer and
  Communications Security ({CCS})}, 2012.

\bibitem{CashJJKRS13}
D.~Cash, S.~Jarecki, C.~Jutla, H.~Krawczyk, M.~Rosu, and M.~Steiner,
  ``{Highly-Scalable Searchable Symmetric Encryption with Support for Boolean
  Queries},'' in \emph{Proceedings of the International Cryptology Conference
  ({CRYPTO})}, 2013.

\bibitem{OgataKKM13}
W.~Ogata, K.~Koiwa, A.~Kanaoka, and S.~Matsuo, ``{Toward Practical Searchable
  Symmetric Encryption},'' in \emph{Proceedings of the International Workshop
  on Security ({IWSec})}, 2013.

\bibitem{LauCSJLB14}
B.~Lau, S.~P. Chung, C.~Song, Y.~Jang, W.~Lee, and A.~Boldyreva, ``{{Mimesis
  Aegis}: {A} Mimicry Privacy Shield - A System's Approach to Data Privacy on
  Public Cloud},'' in \emph{Proceedings of the USENIX Security Symposium},
  2014.

\bibitem{Kurosawa14}
K.~Kurosawa, ``Garbled searchable symmetric encryption,'' in \emph{Proceedings
  of the International Conference on Financial Cryptography and Data Security
  ({FC})}, 2014.

\bibitem{NaveedPG14}
M.~Naveed, M.~Prabhakaran, and C.~A. Gunter, ``{Dynamic Searchable Encryption
  via Blind Storage},'' in \emph{Proceedings of the IEEE Symposium on Security
  and Privacy ({IEEE S\&P})}, 2014.

\bibitem{StefanovPS14}
E.~Stefanov, C.~Papamanthou, and E.~Shi, ``{Practical Dynamic Searchable
  Encryption with Small Leakage},'' in \emph{Proceedings of the Network and
  Distributed System Security Symposium ({NDSS})}, 2014.

\bibitem{HeAJSS14}
W.~He, D.~Akhawe, S.~Jain, E.~Shi, and D.~X. Song, ``{ShadowCrypt: Encrypted
  Web Applications for Everyone},'' in \emph{Proceedings of the ACM Conference
  on Computer and Communications Security ({CCS})}, 2014.

\bibitem{Bost16}
R.~Bost, ``{{$\Sigma$}o{$\varphi$}o{$\varsigma$}: Forward Secure Searchable
  Encryption},'' in \emph{Proceedings of the ACM Conference on Computer and
  Communications Security ({CCS})}, 2016.

\bibitem{IslamKK12}
M.~S. Islam, M.~Kuzu, and M.~Kantarcioglu, ``{Access Pattern Disclosure on
  Searchable Encryption: Ramification, Attack and Mitigation},'' in
  \emph{Proceedings of the Network and Distributed System Security Symposium
  ({NDSS})}, 2012.

\bibitem{CashGPR15}
D.~Cash, P.~Grubbs, J.~Perry, and T.~Ristenpart, ``{Leakage-Abuse Attacks
  Against Searchable Encryption},'' in \emph{Proceedings of the ACM Conference
  on Computer and Communications Security ({CCS})}, 2015.

\bibitem{ZhangKP16}
Y.~Zhang, J.~Katz, and C.~Papamanthou, ``{All Your Queries Are Belong to Us:
  The Power of File-Injection Attacks on Searchable Encryption},'' in
  \emph{Proceedings of the USENIX Security Symposium}, 2016.

\bibitem{LiuZWT14}
C.~Liu, L.~Zhu, M.~Wang, and Y.-A. Tan, ``{Search pattern leakage in searchable
  encryption: Attacks and new construction},'' \emph{Inf. Sci.}, vol. 265, pp.
  176--188, 2014.

\bibitem{AbdelraheemAG17}
M.~A. Abdelraheem, T.~Andersson, and C.~Gehrmann, ``{Inference and
  Record-Injection Attacks on Searchable Encrypted Relational Databases},''
  Cryptology ePrint Archive, Report 2017/024, 2017,
  \url{http://eprint.iacr.org/2017/024}.

\bibitem{GrubbsMNRS16}
P.~Grubbs, R.~McPherson, M.~Naveed, T.~Ristenpart, and V.~Shmatikov,
  ``{Breaking Web Applications Built On Top of Encrypted Data},'' in
  \emph{Proceedings of the ACM Conference on Computer and Communications
  Security ({CCS})}, 2016.

\bibitem{PouliotW16}
D.~Pouliot and C.~V. Wright, ``{The Shadow Nemesis: Inference Attacks on
  Efficiently Deployable, Efficiently Searchable Encryption},'' in
  \emph{Proceedings of the ACM Conference on Computer and Communications
  Security ({CCS})}, 2016.

\bibitem{GiaruadABL17}
M.~Giaruad, A.~{Anzala-Yamajako}, O.~Bernard, and P.~Lafourcase, ``{Practical
  Passive Leakage-Abuse Attacks Against Symmetric Searchable Encryption},''
  Cryptology ePrint Archive, Report 2017/046, 2017,
  \url{http://eprint.iacr.org/2017/046}.

\bibitem{Islam:2014:rangepatterns}
M.~S. Islam, M.~Kuzu, and M.~Kantarcioglu, ``{Inference Attack Against
  Encrypted Range Queries on Outsourced Databases},'' in \emph{Proceedings of
  the ACM Conference on Data and Application Security and Privacy ({CODASPY})},
  2014.

\bibitem{KellarisKNO16}
G.~Kellaris, G.~Kollios, K.~Nissim, and A.~{O'Neill}, ``{Generic Attacks on
  Secure Outsourced Databases},'' in \emph{Proceedings of the ACM Conference on
  Computer and Communications Security ({CCS})}, 2016.

\bibitem{Dautrich:2013}
J.~L. Dautrich, Jr. and C.~V. Ravishankar, ``{Compromising Privacy in Precise
  Query Protocols},'' in \emph{International Conference on Extending Database
  Technology (EDBT)}, 2013.

\bibitem{KennyAccessPattern}
M.-S. Lacharit\'e, B.~Minaud, and K.~G. Paterson, ``Improved reconstruction
  attacks on encrypted data using range query leakage,'' in \emph{Proceedings
  of the IEEE Symposium on Security and Privacy ({IEEE S\&P})}, 2018.

\bibitem{GLMP:attack:2019}
P.~Grubbs, M.-S. Lacharit\'e, B.~Minaud, and K.~G. Paterson, ``Learning to
  reconstruct: Statistical learning theory and encrypted database attacks,'' in
  \emph{Proceedings of the IEEE Symposium on Security and Privacy ({IEEE
  S\&P})}, 2019.

\bibitem{KPT:knnattack:2019}
E.~Kornaropoulos, C.~Papamanthou, and R.~Tamassia, ``Data recovery on encrypted
  databases with k-nearest neighbor query leakage,'' in \emph{Proceedings of
  the IEEE Symposium on Security and Privacy ({IEEE S\&P})}, 2019.

\bibitem{KPT:attack:2020}
------, ``The state of the uniform: Attacks on encrypted databases beyond the
  uniform query distribution,'' in \emph{Proceedings of the IEEE Symposium on
  Security and Privacy ({IEEE S\&P})}, 2020.

\bibitem{GoldreichO96}
O.~Goldreich and R.~Ostrovsky, ``{Software Protection and Simulation on
  Oblivious RAMs},'' \emph{J. {ACM}}, pp. 431--473, 1996.

\bibitem{StefanovDSFRYD13}
E.~Stefanov, M.~van Dijk, E.~Shi, C.~W. Fletcher, L.~Ren, X.~Yu, and
  S.~Devadas, ``{Path {ORAM}: an extremely simple oblivious {RAM} protocol},''
  in \emph{Proceedings of the ACM Conference on Computer and Communications
  Security ({CCS})}, 2013.

\bibitem{PIR:survey}
W.~Gasarch, ``{A survey on private information retrieval},'' in \emph{The
  Computational Complexity Column}, 2007.

\bibitem{Grubbs18:Volume}
P.~Grubbs, M.-S. Lacharit\'e, B.~Minaud, and K.~G. Paterson, ``Pump up the
  volume: Practical database reconstruction from volume leakage on range
  queries,'' in \emph{Proceedings of the ACM Conference on Computer and
  Communications Security ({CCS})}, 2018.

\bibitem{Gui19:Volume}
Z.~Gui, O.~Johnson, and B.~Warinschi, ``Encrypted databases: New volume attacks
  against range queries,'' in \emph{Proceedings of the ACM Conference on
  Computer and Communications Security ({CCS})}, 2019.

\bibitem{Signal}
{Signal}. \url{https://signal.org}.

\bibitem{SignalDiscoveryBlog}
M.~Marlinspike, ``Technology preview: Private contact discovery for signal,''
  \url{https://signal.org/blog/private-contact-discovery/}, 2017.

\bibitem{httprfc}
R.~Fielding and J.~Reschke, ``{Hypertext Transfer Protocol (HTTP/1.1): Message
  Syntax and Routing},'' RFC 7230, 2014.

\bibitem{Kellaris:2017:privsearch}
G.~Kellaris, G.~Kollios, K.~Nissim, and A.~O'Neill, ``{Accessing Data while
  Preserving Privacy},'' 2017, \url{http://arxiv.org/abs/1706.01552}.

\bibitem{Kuzu:2014:privsearch}
M.~Kuzu, M.~S. Islam, and M.~Kantarcioglu, ``{Efficient Privacy-aware Search
  over Encrypted Databases},'' in \emph{Proceedings of the ACM Conference on
  Data and Application Security and Privacy ({CODASPY})}, 2014.

\bibitem{Roche:vORAM}
D.~S. Roche, A.~J. Aviv, and S.~G. Choi, ``{A Practical Oblivious Map Data
  Structure with Secure Deletion and History Independence},'' in
  \emph{Proceedings of the IEEE Symposium on Security and Privacy ({IEEE
  S\&P})}, 2016.

\bibitem{Enron}
Enron email dataset. \url{https://www.cs.cmu.edu/~./enron/}.

\bibitem{Porter}
M.~F. Porter, ``{An Algorithm for Suffix Stripping},'' \emph{Readings in
  Information Retrieval}, pp. 313--316, 1997.

\bibitem{English}
English keywords dataset. \url{https://github.com/dwyl/english-words}.

\bibitem{gmailsize}
Gmail size limits. \url{https://support.google.com/mail/answer/6584}.

\bibitem{cdcdiseases}
{Center for Disease Control and Prevention (CDC): Diseases and Conditions A-Z
  Index}. \url{https://www.cdc.gov/DiseasesConditions}.

\bibitem{STIs}
K.~K. {Holmes et al.}, ``{Sexually Transmitted Diseases},'' \emph{McGraw Hill
  Medical, 4th Ed., NY}, 2008.

\bibitem{mitmproxy}
{MITM Proxy}. \url{http://mitmproxy.org/}.

\bibitem{Naveed15}
M.~Naveed, ``{The Fallacy of Composition of Oblivious RAM and Searchable
  Encryption},'' Cryptology ePrint Archive, Report 2015/668, 2015,
  \url{http://eprint.iacr.org/2015/668}.

\bibitem{gilbert2002brewer}
S.~Gilbert and N.~Lynch, ``Brewer's conjecture and the feasibility of
  consistent, available, partition-tolerant web services,'' \emph{ACM SIGACT
  News}, vol.~33, no.~2, pp. 51--59, 2002.

\bibitem{Smart:survey}
N.~P. {Smart (Editor)}, ``{Future Directions in Computing on Encrypted Data},''
  in \emph{{ECRYPT} report}, 2015.

\bibitem{Stefanov:2013}
E.~Stefanov and E.~Shi, ``{ObliviStore: High Performance Oblivious Cloud
  Storage},'' in \emph{Proceedings of the IEEE Symposium on Security and
  Privacy ({IEEE S\&P})}, 2015.

\bibitem{Bindschaedler:2015}
V.~Bindschaedler, M.~Naveed, X.~Pan, X.~Wang, and Y.~Huang, ``{Practicing
  Oblivious Access on Cloud Storage: The Gap, the Fallacy, and the New Way
  Forward},'' in \emph{Proceedings of the ACM Conference on Computer and
  Communications Security ({CCS})}, 2015.

\bibitem{SahinZALT16}
C.~Sahin, V.~Zakhary, A.~{El Abbadi}, H.~Lin, and S.~Tessaro, ``{TaoStore:
  Overcoming Asynchronicity in Oblivious Data Storage},'' in \emph{Proceedings
  of the IEEE Symposium on Security and Privacy ({IEEE S\&P})}, 2016.

\bibitem{mishra:oblix}
P.~Mishra, R.~Poddar, J.~Chen, A.~Chiesa, and R.~A. Popa, ``{Oblix: An
  Efficient Oblivious Search Index},'' in \emph{Proceedings of the IEEE
  Symposium on Security and Privacy ({IEEE S\&P})}, 2018.

\bibitem{SQLQuery:PIR:Goldberg}
F.~Olumofin and I.~Goldberg, ``{Privacy-preserving queries over relational
  databases},'' in \emph{Proceedings of the Privacy Enhancing Technologies
  Symposium ({PETS})}, 2010.

\bibitem{PIRPublicData}
S.~Wang, D.~Agrawal, and A.~E. Abbadi, ``{Generalizing PIR for Practical
  Private Retrieval of Public Data},'' in \emph{IFIP Annual Conference on Data
  and Applications Security and Privacy}, 2010.

\bibitem{PIRBook}
D.~Asonov, \emph{{Querying Databases Privately: {A} New Approach to Private
  Information Retrieval}}, ser. Lecture Notes in Computer Science.\hskip 1em
  plus 0.5em minus 0.4em\relax Springer, 2004, vol. 3128.

\bibitem{Blackstone19}
L.~Blackstone, S.~Kamara, and T.~Moataz, ``{Revisiting Leakage Abuse
  Attacks},'' in \emph{Proceedings of the Network and Distributed System
  Security Symposium ({NDSS})}, 2020.

\end{thebibliography}
}

\appendices %

\section{Extensions for conjunctive queries}
\label{s:conjunctions}

In this section, we describe extensions to our base attack for identifying keywords in conjunctive queries. 
The extensions are based on attacks described by Zhang \etal~\cite{ZhangKP16}. 
Zhang \etal consider powerful attackers who can observe file access patterns and thereby uniquely identify documents in the result set. 
Instead, we adapt the attacks for our setting in which the attacker observes nothing more than the size of the result set. 
We present two attacks---the first optimizes the number of required replays while the second reduces the number of files injected.

\subsection{Reducing the number of required replays}
Zhang \etal~\cite{ZhangKP16} present a general attack for conjunctive queries with $d$ keywords. 
The attacker injects $n$ files into the database, each containing $L$ randomly chosen keywords from the dictionary $D$. 
They claim that for properly chosen $n$ and $L$, the intersection of the returned files will contain exactly the $d$ queried keywords and no others with a very large probability. 
The authors prove the claim for $L = \frac{1}{2}^{1/d}|D|$, and $n = (2+\epsilon)d \log |D|$ (where $\epsilon > 0$) and show that the probability of success in this case is $e^{(-1/|D|)^{\epsilon/4}}$.  

We extend the above attack as follows. 
The attacker creates the $n$ files as before, but injects $2^i$ copies of the $i$-th file into the database (for $i \in [0, n-1]$). 
The number of files returned is thus sufficient to uniquely identify the exact subset of files whose copies were returned in the result.

The attack only requires a single replay of the query, but the total number of files injected by the attacker in this case is equal to $2^0 + 2^1 + \ldots + 2^{n-1} \approx 2^n = |D|^{(2+\epsilon)d}$. The attack is thus more suited for situations with small dictionaries.

\subsection{Reducing the number of files injected}
Zhang \etal also present an adaptive attack for conjunctive queries with $d$ keywords, which reduces the number of files the attacker needs to inject. The core idea is to first perform a binary search in order to identify the lexicographically largest keyword $w$ in the query. Once $w$ is identified, the attacker performs another binary search to identify the next keyword in the query, but with $w$ present in all the injected files. The attacker proceeds in this manner until all the keywords are identified.

Specifically, the attacker orders all keywords in the dictionary lexicographically, and then injects a single file containing the first $|D|/2$ keywords. 
If the size of the result set increases by one (\ie the response includes the injected file), then he repeats the attack by injecting another file containing the first $|D|/4$ keywords; on the other hand, if the response does not include the file, then he injects a file containing the first $3|D|/4$ keywords instead. 
The attacker repeats the process $\log|D|$ times, until the first (lexicographically largest) keyword is identified. 
The attack applies straightforwardly in our setting.

For this variant of the attack, both the number of required replays and the total number of files injected are equal to $d \log |D|$. Compared to the attack in previous section, this variant drastically reduces the number of files that need to be injected, but also increases the number of required replays.

\end{document}